\documentclass[11pt]{article}
\usepackage{amssymb,latexsym,amsmath,amsbsy,amsthm}
\usepackage[dvips]{graphicx}
\usepackage{cite}
\usepackage{stmaryrd}  
\usepackage{arydshln}  
\usepackage{mathdots}  
\newtheorem{theo}{Theorem}
\newtheorem{defi}[theo]{Definition}

\newtheorem{prop}[theo]{Proposition}
\def\nn{\nonumber}
\def\deg{\mathop{\rm deg}\nolimits}

\def\qdots{\mathinner{\mkern1mu\raise1pt\vbox{\kern7pt\hbox{.}}\mkern2mu \raise4pt\hbox{.}\mkern2mu\raise7pt\hbox{.}\mkern1mu}}
\def\Z{{\mathbb Z}}

\def\gl{\mathfrak{gl}}
\def\ssl{\mathfrak{sl}}

\def\g{\mathfrak{g}}
\def\h{\mathfrak{h}}
\def\so{\mathfrak{so}}

\def\osp{\mathfrak{osp}}
\def\pso{\mathfrak{pso}}
\def\B{\mathfrak{B}}
\def\lb{\llbracket}
\def\rb{\rrbracket}

\def\beq{\begin{equation}}
\def\eeq{\end{equation}}
\newcommand{\overbar}[1]{\mkern 3.0mu\overline{\mkern-3.0mu#1\mkern-3.0mu}\mkern 3.0mu}

\renewcommand{\theequation}{\arabic{section}.{\arabic{equation}}}

\begin{document}
\begin{center}
{\Large \bf
Representations of the Lie superalgebra $\B(\infty,\infty)$\\[2mm] 
and parastatistics Fock spaces} \\[5mm]
{\bf N.I.~Stoilova}\footnote{E-mail: stoilova@inrne.bas.bg}\\[1mm] 
Institute for Nuclear Research and Nuclear Energy,\\ 
Boul.\ Tsarigradsko Chaussee 72, 1784 Sofia, Bulgaria\\[2mm] 
{\bf J.\ Van der Jeugt}\footnote{E-mail: Joris.VanderJeugt@UGent.be}\\[1mm]
Department of Applied Mathematics, Computer Science and Statistics, Ghent University,\\
Krijgslaan 281-S9, B-9000 Gent, Belgium.
\end{center}


\begin{abstract}
The algebraic structure generated by the creation and annihilation operators of a system of $m$ parafermions and $n$ parabosons, satisfying the mutual parafermion relations, is known to be the Lie superalgebra $\osp(2m+1|2n)$.
The Fock spaces of such systems are then certain lowest weight representations of $\osp(2m+1|2n)$.
In the current paper, we investigate what happens when the number of parafermions and parabosons becomes infinite.
In order to analyze the algebraic structure, and the Fock spaces, we first need to develop a new matrix form for the Lie superalgebra $\B(n,n)=\osp(2n+1|2n)$, and construct a new Gelfand-Zetlin basis of the Fock spaces in the finite rank case.
The new structures are appropriate for the situation $n\rightarrow\infty$.
The algebra generated by the infinite number of creation and annihilation operators is $\B(\infty,\infty)$, a well defined infinite rank version of the orthosymplectic Lie superalgebra. 
The Fock spaces are lowest weight representations of $\B(\infty,\infty)$, with a basis consisting of particular row-stable Gelfand-Zetlin patterns. 
\end{abstract}



\setcounter{equation}{0}
\section{Introduction} \label{sec:Introduction}%

The original motivation for the present paper comes from some classical physical ideas, and more precisely from some conceptual difficulties of quantum mechanics.  
Let us sketch some of the historical background, starting with Wigner's paper ``Do the equations of motion determine the quantum mechanical commutation relations?''~\cite{Wigner}. 
On the simplest example, the one dimensional harmonic oscillator, Wigner showed that a more general approach is to start from the equations of motion, Hamilton's equations and the Heisenberg equations, instead of assuming that the position and momentum operators are subject to the canonical commutation relations. 
Actually he found an infinite set of solutions of the compatibility conditions of Hamilton's equations and the Heisenberg equations.
One of these solutions coincides with the canonical commutation relations. 
Each of these solutions is now known to correspond to a representation of the Lie superalgebra $\osp(1|2)$. 

Wigner's solutions were generalized in 1953 by Green who wrote down the so called paraboson relations~\cite{Green}, generalizing the quadratic relations for Bose operators. 
He also generalized the ordinary Fermi operators to parafermion operators~\cite{Green}. 
Parabosons and parafermions were applied in particular to quantum field theory~\cite{Haag, Wu, Ohnuki} and to generalizations of quantum statistics~\cite{Green, GM, Tolstoy2014, YJ, YJ2, KD, Nelson, Kitabayashi, Huerta2017, Huerta2018 } (parastatistics).

We have been interested in parabosons and parafermions primarily because of the underlying algebraic structures, and the relation between Fock spaces and representations of these algebraic structures.
Let us make this concrete by recalling some of the known algebraic relations.
Whereas fermion operators satisfy quadratic (anti-commutation) relations, parafermion operators satisfy certain triple relations: a system of $m$ parafermion creation and annihilation operators $f_j^\pm$ ($j=1,\ldots,m$) is determined by
\begin{equation}
[[f_{ j}^{\xi}, f_{ k}^{\eta}], f_{l}^{\epsilon}]=
|\epsilon -\eta| \delta_{kl} f_{j}^{\xi} - |\epsilon -\xi| \delta_{jl}f_{k}^{\eta}, 
\label{f-rels}
\end{equation}
where $j,k,l\in \{1,2,\ldots,m\}$ and $\eta, \epsilon, \xi \in\{+,-\}$ (to be interpreted as $+1$ and $-1$
in the algebraic expressions $\epsilon -\xi$ and $\epsilon -\eta$).
Similarly, a system of $n$ pairs of parabosons $b_j^\pm$ satisfies
\begin{equation}
[\{ b_{ j}^{\xi}, b_{ k}^{\eta}\} , b_{l}^{\epsilon}]= 
(\epsilon -\xi) \delta_{jl} b_{k}^{\eta}  + (\epsilon -\eta) \delta_{kl}b_{j}^{\xi}.
\label{b-rels}
\end{equation}
These cubic or triple relations involve nested (anti-)commutators, just like the Jacobi identity of Lie (super)algebras.
It was indeed shown later~\cite{Kamefuchi,Ryan} that the parafermionic algebra determined by~\eqref{f-rels} is 
the orthogonal Lie algebra $\mathfrak{so}(2m+1)$, and that the parabosonic algebra determined by~\eqref{b-rels} is the orthosymplectic Lie superalgebra $\mathfrak{osp}(1|2n)$~\cite{Ganchev}. 

Greenberg and Messiah~\cite{GM} considered combined systems of parafermions and parabosons. 
Apart from two trivial combinations, there are two non-trivial relative commutation relations between parafermions and parabosons, also expressed by means of triple relations. 
The first of these are the so-called relative parafermion relations, determined by:
\begin{align}
&[[f_{ j}^{\xi}, f_{ k}^{\eta}], b_{l}^{\epsilon}]=0,\qquad [\{b_{ j}^{\xi}, b_{ k}^{\eta}\}, f_{l}^{\epsilon}]=0, \nn\\
&[[f_{ j}^{\xi}, b_{ k}^{\eta}], f_{l}^{\epsilon}]= -|\epsilon-\xi| \delta_{jl} b_k^{\eta}, \qquad
\{[f_{ j}^{\xi}, b_{ k}^{\eta}], b_{l}^{\epsilon}\}= (\epsilon-\eta) \delta_{kl} f_j^{\xi}.
\label{rel-pf}
\end{align}
The parastatistics algebra with relative parafermion relations, determined by~\eqref{f-rels}, 
\eqref{b-rels} and~\eqref{rel-pf}, was identified by
Palev~\cite{Palev1} and is the orthosymplectic Lie superalgebra $\mathfrak{osp}(2m+1|2n)$.
The second case, where~\eqref{f-rels} and~\eqref{b-rels} are combined with so-called relative paraboson relations, leads to an algebra which has received attention in a number of papers~\cite{YJ,YJ2,KA,KK,Tolstoy2014}, and is no longer a Lie superalgebra but a $\Z_2\times\Z_2$-graded algebraic structure~\cite{Tolstoy2014,SV2018}.

The identification of the underlying algebraic structures is one important aspect.
The other important problem is to identify and describe the paraboson/parafermion Fock spaces as representations of the algebras involved. 
It is to this problem that we have contributed a number of solutions during the last years.
These Fock spaces are characterized by a positive integer~$p$, often referred to as the order of statistics.
Although there have been many approaches to parastatistics Fock spaces (often based on the so-called Green's ansatz), the construction of a complete orthogonal basis of the Fock space and explicit actions of the parastatistics operators on these basis vectors has been completed only fairly recently.
The explicit Fock representations for a system of $m$ parafermions was given in~\cite{parafermion}, and that for a system of $n$ parabosons in~\cite{paraboson}.
In both of these cases, the analysis of the Fock space by means of a Lie subalgebra of type $\gl$ of $\so(2m+1)$ or $\osp(1|2n)$ was essential, since that subalgebra provided a Gelfand-Zetlin (GZ) basis for the Fock space.
In~\cite{SV2015}, we managed to extend these results to a combined system of $m$ parafermions and $n$ parabosons, satisfying the mutual relations~\eqref{rel-pf}. For this case, the Lie superalgebra $\gl(m|n)$ as subalgebra of $\osp(2m+1|2n)$ was essential, providing a GZ basis of the Fock spaces, being certain infinite-dimensional lowest weight representations of $\osp(2m+1|2n)$.

If one thinks of quantum fields described by parafermions and parabosons, one should consider however an infinite number of parabosons and parafermions~\cite{Palev92}.
Also from the mathematical point of view, letting $m$ and $n$ go to $+\infty$ raises many interesting questions.
For an infinite set of parafermions only, the underlying algebra becomes the infinite rank Lie algebra $\so(\infty)$, whereas for an infinite set of parabosons the underlying algebra becomes the infinite rank Lie superalgebra $\osp(1|\infty)$.
Both of these cases have been studied, and for the Fock space it was possible to extend the GZ-basis of the finite rank algebras to the infinite rank case~\cite{SV2009}.
For a combined system of an infinite number of parafermions and an infinite number of parabosons, the problem of identifying the underlying algebra and the structure of the Fock spaces remained open.
This is in fact the problem being solved in the current paper.

Note that this problem is not a straightforward limit generalization of the solution described in~\cite{SV2015}.
Indeed, for the $\osp(2m+1|2n)$ solution, the basis of the Fock space consists of GZ-patterns related to a subalgebra chain of the following type:
\begin{align}
&\osp(2m+1|2n) \supset \mathfrak{gl}(m|n) \supset \mathfrak{gl}(m|n-1) \supset \mathfrak{gl}(m|n-2) \supset \cdots \nn\\
&\supset \mathfrak{gl}(m|1) \supset \mathfrak{gl}(m) \supset \mathfrak{gl}(m-1) \supset\cdots\supset \mathfrak{gl}(2) \supset \mathfrak{gl}(1).
\label{chain-mn}
\end{align}
As already discussed in~\cite{SV2016}, the corresponding GZ-basis patterns could possibly be extended for $n$ going to $\infty$, but not for both $m$ and $n$ going to infinity, since the above chain of subalgebras does not provide proper GZ-basis vectors for $\gl(\infty|\infty)$.
Instead, one has to introduce a different GZ-basis, referred to as the ``odd Gelfand-Zetlin basis'', and constructed for covariant representations of $\gl(n|n)$ in~\cite{SV2016}.
Note that one had to take $m=n$ in order to construct this basis. 
So in the case under consideration, the relevant chain of subalgebras is
\begin{align}
& \osp(2n+1|2n)\supset \mathfrak{gl}(n|n) \supset \mathfrak{gl}(n|n-1) \supset \mathfrak{gl}(n-1|n-1) \supset \mathfrak{gl}(n-1|n-2) \supset \mathfrak{gl}(n-2|n-2) \supset \cdots \nn\\
& \supset \gl(2|2) \supset \gl(2|1) \supset \mathfrak{gl}(1|1) \supset \mathfrak{gl}(1).
\label{chain-n=n}
\end{align}
The ``odd'' GZ-patterns for the Fock spaces, as representations of $\osp(2n+1|2n)$, derived from this chain will have certain stability properties that allow us to let $n$ grow to infinity.

In the present paper, we therefore start from the Lie superalgebra $\B(n,n)=\osp(2n+1|2n)$, for which we provide a new matrix realization in Section~\ref{sec:B} (as the common one cannot be extended to infinite-dimensional matrices).
In this new matrix realization, the operators corresponding to $n$ parafermions and $n$ parabosons are identified, and seen to generate a basis for $\osp(2n+1|2n)$.
The Fock space of order $p$ for such a set of parastatistics operators is identified as a lowest weight representation $V(p,n)$ of $\osp(2n+1|2n)$ in Section~\ref{sec:C}.
The main difference with~\cite{SV2015} is that now the basis vectors of $V(p,n)$ are given in another form, namely as certain odd GZ-patterns of $\gl(n|n)$.
Since we are dealing with a new basis, the actions of the parastatistics operators also needs to be recomputed. 
This computation is really involved, although it follows the same line of thought as in~\cite{SV2015}, namely the matrix elements of the parastatistics operators should be products of $\gl(n|n)$ Clebsch-Gordan coefficients and reduced matrix elements.
The reduced matrix elements are basis-independent, so they must coincide with the ones computed in~\cite{SV2015}.
But the $\gl(n|n)$ Clebsch-Gordan coefficients depend on the basis, so we had to recompute them in this ``odd'' GZ-basis. 
The results of this computation is given separately in Appendix~\ref{A}.
In Section~\ref{sec:D}, we study the basis vectors of the Fock space $V(p,n)$ of $\B(n,n)$ in further detail.
The integer entries in the GZ-patterns of the basis vectors turn out to have interesting (combinatorial) properties.
They have many stability properties that are preserved under the action of the parastatistics creation and annihilation operators.
These stability properties allow us to define GZ-patterns for the case that $n$ becomes infinite.
In Section~\ref{sec:E}, the infinite rank Lie superalgebra $\B(\infty,\infty)$ is defined by means of a matrix form, consisting of certain infinite square matrices with only a finite number of nonzero entries.
The identification of $\B(\infty,\infty)$ as the Lie superalgebra generated by an infinite number of parafermions and parabosons (subject to particular mutual relations) is then rather straightforward.
Then we turn to the Fock spaces $V(p)$ of such combined systems of parafermions and parabosons.
Our analysis shows that a basis for these Fock spaces is given by infinite but row-stable odd GZ-patterns.
The action of the parastatistics creation and annihilation operators on these basis vectors is given, and the rest of the section is devoted to proving that $V(p)$ is indeed an irreducible representation of $\B(\infty,\infty)$ under the given action.
The underlying idea is that the row-stability of the infinite GZ-patterns allow one to extend operator actions valid for the finite rank case to the case that $n$ becomes infinite.

The current paper forms the closing piece of a number of studies in which Fock representations were explicitly constructed, chronologically for a set of $n$ parabosons~\cite{paraboson}, a set of $m$ parafermions~\cite{parafermion}, an infinite set of parabosons or parafermions~\cite{SV2009}, a combined set of $m$ para\-fermions and $n$ parabosons~\cite{SV2015,SV2018}, and now a combined set of an infinite number of parafermions and parabosons.

\section{The Lie superalgebras $\B(n|n)\equiv \osp(2n+1|2n)$}
\setcounter{equation}{0} \label{sec:B}

In order to extend the results and the notation of the present section to the case of an infinite rank Lie superalgebra, it will be necessary to work with a different matrix realization (over the complex numbers) of the Lie superalgebra $\B(n|n)\equiv \osp(2n+1|2n)$ than the commonly one used in~\cite{Kac,SV2015}. 
We will give this new matrix realization here for $\B(n,n)$, since that is the only case we need, but it is clear that it could be given for the more general case $\B(m,n)$.
For the labelling of rows and columns of matrices (and other objects), it will be convenient to use both negative and positive integers. 
In particular, when $m$ and $n$ are non-negative integers, we will use the following notation for ordered sets:
\beq
[-m,n]=\{-m,\ldots,-2,-1,0,1,2,\ldots,n\}, \qquad [-m,n]^*=\{-m,\ldots,-2,-1,1,2,\ldots,n\}.
\eeq
Sometimes it will be convenient to write the minus sign of an index as an overlined number. 
With this convention, we have e.g.
\[
[\bar{2},3]^*=\{\bar{2},\bar{1},1,2,3\}=\{-2,-1,1,2,3\} \qquad\hbox{and}\qquad 
[\bar{n},\bar{1}]=\{\bar{n},\ldots,\bar{2},\bar{1}\}=\{-n,\ldots,-2,-1\}.
\]
We will also use
\[
\Z^*=\Z\setminus\{0\},\qquad \Z_+=\{0,1,2,\ldots\},\qquad \Z_+^*=\{1,2,3,\ldots\}
\]
and similarly for $\Z_-$ and $\Z_-^*$.

Let $I$ and $J$ be the $(2\times 2)$-matrices
\beq
I:=\left(\begin{array}{cc} 0&1 \\ 1&0 \end{array}\right), \qquad
J:=\left(\begin{array}{cc} 0&1 \\ -1&0 \end{array}\right), 
\eeq
and let $B$ be the $(4n+1)\times(4n+1)$-matrix, with indices in $[-2n,2n]$, given by $B=I\oplus \cdots \oplus I \oplus 1 \oplus J \oplus \cdots \oplus J$, or, written in block form:
\beq
B:=
\left( \begin{array}{ccc:c:ccc}
I & 0 &&&&& \\
0 & \ddots & 0 &&&&\\
 & 0 & I &&&&\\
\hdashline
 & & & 1 & & & \\
\hdashline
&&&& J & 0 & \\
&&&& 0 & \ddots & 0\\
&&&& & 0 & J
\end{array} \right).
\eeq
Herein, 0 stands for the zero $(2\times2)$-matrix, the entry 1 is at position $(0,0)$, and the empty parts of the matrix consist of zeros.

The matrices $X$ of the Lie superalgebra $\B(n,n)$ will have the following block form:
\beq
X:=
\left( \begin{array}{ccc:c|ccc}
X_{\bar{n},\bar{n}} & \cdots & X_{\bar{n},\bar{1}} & X_{\bar{n},0} & X_{\bar{n},1} & \cdots & X_{\bar{n},n}\\
\vdots & \ddots & \vdots & \vdots & \vdots & \ddots & \vdots \\
X_{\bar{1},\bar{n}} & \cdots & X_{\bar{1},\bar{1}} & X_{\bar{1},0} & X_{\bar{1},1} & \cdots & X_{\bar{1},n}\\
\hdashline
X_{{0},\bar{n}} & \cdots & X_{{0},\bar{1}} & 0 & X_{{0},1} & \cdots & X_{{0},n}\\
\hline
X_{{1},\bar{n}} & \cdots & X_{{1},\bar{1}} & X_{{1},0} & X_{{1},1} & \cdots & X_{{1},n}\\
\vdots & \ddots & \vdots & \vdots & \vdots & \ddots & \vdots \\
X_{{n},\bar{n}} & \cdots & X_{{n},\bar{1}} & X_{{n},0} & X_{{n},1} & \cdots & X_{{n},n}
\end{array} \right).
\label{X}
\eeq
Herein, any matrix of the form $X_{ij}$ with $i,j \in[\bar{n},n]^*$ is a $(2\times 2)$-matrix, $X_{0,i}$ is a $(1\times 2)$-matrix and $X_{i,0}$ a $(2\times 1)$-matrix.

The Lie superalgebra $\g=\B(n,n)=\osp(2n+1|2n)$ is $\Z_2$-graded, $\g=\g_0\oplus\g_1$, whose homogeneous elements are referred to as even and odd elements, and the degree of a homogeneous element $X$ is denoted by $\deg(X)$.
The even matrices $X$ will have zeros in the upper right and bottom left blocks, i.e.\ $X_{ij}=0$ for all $(i,j)\in[\bar{n},0]\times[1,n]$ and $(i,j)\in[1,n]\times[\bar{n},0]$.
The odd matrices $X$ will have zeros in the upper left and bottom right blocks, i.e.\ $X_{ij}=0$ for all $(i,j)\in[\bar{n},0]\times[\bar{n},0]$ and $(i,j)\in[1,n]\times[1,n]$.

The actual definition, derived from~\cite{Kac}, is then as follows: $\B(n,n)_0$ consists of all even matrices $X$ of the form~\eqref{X} such that 
\[
X^TB+BX=0;
\]
$\B(n,n)_1$ consists of all odd matrices $X$ of the form~\eqref{X} such that 
\[
X^{ST}B-BX=0.
\]
Herein $X^T$ is the ordinary transpose of $X$; $X^{ST}$ is the supertranspose of $X$, which is, for an odd matrix of the form $X=\left(\begin{array}{c|c} 0 &U\\ \hline V&0\end{array}\right)$ given by $X^{ST}=\left(\begin{array}{c|c} 0 & V^T\\ \hline -U^T&0\end{array}\right)$.

Concretely, $X$ belongs to $\B(n,n)$ provided the blocks of~\eqref{X} satisfy:
\begin{align*}
& I X_{\bar{i},\bar{j}}+X_{\bar{j},\bar{i}}^T I=0, \quad
JX_{{i},{j}}+X_{{j},{i}}^T J=0, \quad
I X_{\bar{i},{j}}-X_{{j},\bar{i}}^T J=0 \qquad(i,j\in[1,n]);\\
& X_{0,\bar{j}}+X_{\bar{j},0}^T I=0, \quad X_{0,{j}}-X_{{j},0}^T  J=0 \qquad(j\in[0,n]).
\end{align*}

For homogeneous elements of type~\eqref{X}, the Lie superalgebra bracket is 
\[
\lb X, Y \rb = XY - (-1)^{\deg(X)\deg(Y)}YX,
\]
with ordinary matrix multiplication in the right hand side.

Denote, as usual, by $e_{ij}$ the matrix with zeros everywhere except a $1$ on position $(i,j)$, where the row and column indices run from $-2n$ to $2n$.
A basis of the Cartan subalgebra $\h$ of $\B(n,n)$ consists of the elements $h_i=e_{2i-1,2i-1}-e_{2i,2i}$ ($i\in[1,n]$) and
$h_{i}=e_{2i,2i}-e_{2i+1,2i+1}$ ($i\in[\bar{n},\bar{1}]$). 
The corresponding dual basis of $\h^*$ will be denoted by $\epsilon_i$ ($i\in[\bar{n},n]^*$).
The following elements are even root vectors with roots $\epsilon_{-i}$ and $-\epsilon_{-i}$ respectively ($i\in[1,n]$):
\begin{align}
&c_{-i}^+\equiv f_{-i}^+= \sqrt{2}(e_{-2i,0}-e_{0,-2i+1}), \nn\\
&c_{-i}^-\equiv f_{-i}^-= \sqrt{2}(e_{0,-2i}-e_{-2i+1,0}), \;
\label{f-as-e}
\end{align}
and odd root vectors with roots $\epsilon_{i}$ and $-\epsilon_{i}$ respectively ($i\in[1,n]$) are given by:
\begin{align}
&c_{i}^+ \equiv b_{i}^+= \sqrt{2}(e_{0,2i}+e_{2i-1,0}), \nn\\
&c_{i}^-\equiv b_{i}^-= \sqrt{2}(e_{0,2i-1}-e_{2i,0}). \; 
\label{b-as-e}
\end{align}
The operators $c_i^+$ are positive root vectors, and the $c_i^-$ are negative root vectors, their $\Z_2$ grading being
\[
\deg(c_i^\pm)=0 \hbox{ for }i\in[\bar{n},\bar{1}], \qquad
\deg(c_i^\pm)=1 \hbox{ for }i\in[1,n].
\]
The remaining root vectors of $\B(n,n)$ are given by elements of the form $\lb c^\xi_i, c^\eta_j \rb$, where $\lb \cdot, \cdot \rb$ is the Lie superalgebra bracket (corresponding to a commutator or anti-commutator).
The factor $\sqrt{2}$ in~\eqref{f-as-e}-\eqref{b-as-e} is introduced because then these operators satisfy the triple relations of parafermion and paraboson creation and annihilation operators (subject to the relative parafermion relations), given in the following theorem~\cite{Palev1}.

\begin{theo}[Palev]
\label{theo1}
As a Lie superalgebra defined by generators and relations, 
$\B(n,n)$ is generated by the elements $c_j^\pm$ ($j\in[\bar{n},n]^*$) subject to the following relations
\begin{align}
& [[f_{ j}^{\xi}, f_{ k}^{\eta}], f_{l}^{\epsilon}]=
|\epsilon -\eta| \delta_{kl} f_{j}^{\xi} - |\epsilon -\xi| \delta_{jl}f_{k}^{\eta}, \qquad (j,k,l\in[\bar{n},\bar{1}];
\label{f-rel}\\
& [\{ b_{ j}^{\xi}, b_{ k}^{\eta}\} , b_{l}^{\epsilon}]= 
(\epsilon -\xi) \delta_{jl} b_{k}^{\eta}  + (\epsilon -\eta) \delta_{kl}b_{j}^{\xi}, \qquad (j,k,l\in[1,n]);
\label{b-rel}\\
& [[f_{ i}^{\xi}, f_{ j}^{\eta}], b_{k}^{\epsilon}]=0,\qquad [\{b_{ k}^{\xi}, b_{ l}^{\eta}\}, f_{i}^{\epsilon}]=0, \nn\\
& [[f_{ j}^{\xi}, b_{ k}^{\eta}], f_{i}^{\epsilon}]= -|\epsilon-\xi| \delta_{ji} b_k^{\eta}, \qquad
\{[f_{ j}^{\xi}, b_{ k}^{\eta}], b_{l}^{\epsilon}\}= (\epsilon-\eta) \delta_{kl} f_j^{\xi}, 
\qquad (i,j\in[\bar{n},\bar{1}],\ k,l\in[1,n]).
\label{fb-rel}
\end{align} 
Herein, $\eta, \epsilon, \xi \in\{+,-\}$, are interpreted as $+1$ and $-1$
in the algebraic expressions $\epsilon -\xi$ and $\epsilon -\eta$.
\end{theo}
The so-called triple relations~\eqref{f-rel}-\eqref{fb-rel} are important: they combine a system of $n$ pairs of parafermion operators with a system of $n$ pairs of paraboson operators in a particular way, that can be extended when $n$ goes to infinity.

Before turning to representations, let us identify a subalgebra of $\B(n,n)$ that will play an important role.

\begin{prop}
The $4n^2$ elements 
\begin{equation}
\lb c_i^+, c_j^-\rb \qquad(i,j\in[\bar{n},n]^*)
\label{un}
\end{equation} 
are a basis of the subalgebra $\gl(n|n)$.
\end{prop}
This follows immediately from the fact that the elements
\[
E_{ij}=\frac12 \lb c_j^+, c_k^- \rb
\]
satisfy
\beq
\lb E_{ij}, E_{kl} \rb = \delta_{jk} E_{il} - (-1)^{\deg(E_{ij})\deg(E_{kl})} \delta_{ij} E_{kj}.
\eeq
Note also that 
\begin{equation}
[c_i^+,c_i^-]=2 h_i\quad (i\in[\bar{n},\bar{1}]), \qquad \{c_{i}^+,c_{i}^-\}=2 h_{i}\quad (i\in[1,n]).
\label{Cartan-h}
\end{equation}
Hence $\h=\hbox{span}\{h_i,\ i\in[\bar{n},n]^*\}$, the Cartan subalgebra of $\B(n,n)$, is also the Cartan subalgebra of $\gl(n|n)$.

\section{Fock representations of $\B(n,n)$ in the ``odd basis''}
\setcounter{equation}{0} \label{sec:C}

The parastatistics Fock space of order $p$ (for the relative parafermion relations), with $p$ a positive integer, has been constructed before~\cite{SV2015} as an infinite-dimensional lowest weight representation $V(p,n)$ of the algebra $\B(n,n)$.  
By definition~\cite{GM,Palev1} the parastatistics Fock space $V(p,n)$ is the Hilbert space with vacuum vector $|0\rangle$, 
defined by means of 
\begin{align}
& \langle 0|0\rangle=1, \qquad c_{j}^- |0\rangle = 0, 
\qquad (c_j^\pm)^\dagger = c_j^\mp, \qquad
\lb c_j^-, c_k^+ \rb |0\rangle = p\delta_{jk}\, |0\rangle \quad
(j,k\in[\bar{n},n]^*)
\label{Fock}
\end{align}
and which is irreducible  under the action of the algebra $\B(n,n)$ generated by the elements $c_j^\pm$. 

By~\eqref{Cartan-h}, one sees that $|0\rangle$ is the lowest weight vector of $V(p,n)$ with weight $[-\frac{p}{2},\ldots, -\frac{p}{2};\frac{p}{2},\ldots, \frac{p}{2}]$ in the basis $\{\epsilon_{-n},\ldots,\epsilon_{-1};\epsilon_1,\ldots,\epsilon_n\}$.
These representations have been analyzed in~\cite{SV2015}.
The main result is the decomposition with respect to the subalgebra chain $\B(n,n)\supset \gl(n|n)$, because then the Gelfand-Zetlin  basis of the $\gl(n|n)$ representations can be used to label the vectors of $V(p,n)$.
The main difference here, compared to~\cite{SV2015}, is that we will need to use a different GZ-basis for the $\gl(n|n)$ representations, one that is appropriate for letting $n$ grow to infinity.
But of course, the branching $\B(n,n)\supset \gl(n|n)$ is the same, and thus we have~\cite{SV2015}:

\begin{prop}
In the decomposition of $V(p,n)$ with respect to $\B(n,n)\supset \gl(n|n)$, all covariant representations of $\gl(n|n)$ labeled by a partition $\lambda=(\lambda_1,\lambda_2,\ldots)$ appear with multiplicity~1, subject to $\lambda_1\leq p$ and $\lambda_{n+1}\leq n$.
In the basis $\{\epsilon_{-n},\ldots,\epsilon_{-1};\epsilon_1,\ldots,\epsilon_n\}$, the highest weight of the $\gl(n|n)$ covariant representation labeled by $\lambda$ is given by~\cite{JHKR} 
\begin{align}
{}[m]^{2n}& =[m_{-n,2n},\ldots,m_{-2,2n},m_{-1,2n};m_{1,2n},m_{2,2n}, \ldots,m_{n,2n}]\nonumber\\
& = [m_{\bar n,2n},\ldots,m_{\bar 2,2n},m_{\bar 1,2n};m_{1,2n},m_{2,2n}, \ldots,m_{n,2n}],
\label{components}
\end{align}
where
\begin{align}
& m_{-i,2n}= \lambda_{n-i+1}, \qquad (i\in[1,n]) \label{mneg}\\
& m_{i,2n} = \max \{ 0, {\lambda_i'}-n \}, \qquad (i\in [1,n]), \label{mpos}
\end{align}
and $\lambda'$ is the partition conjugate to $\lambda$ (for partitions we follow the standard notations of~\cite{Macdonald}).
\end{prop}

Note that~\eqref{mneg}-\eqref{mpos} implies
\begin{equation}
m_{i,2n}-m_{i+1,2n}\in {\mathbb Z}_+, \qquad(i\in[\bar{n},\bar{2}]\cup[1,n-1])
\label{cond1}
\end{equation}
and 
\begin{equation}
m_{-1,2n}\geq \# \{i:m_{i,2n}>0,\; i\in[1,n] \}. \label{cond2}
\end{equation}

An appropriate basis for the vectors of these covariant $\gl(n|n)$ modules has been given in~\cite{SV2016}, and is referred to as ``the odd GZ-basis.''
This corresponds to a decomposition according to the subalgebra chain
\begin{equation}
\mathfrak{gl}(n|n) \supset \mathfrak{gl}(n|n-1) \supset \mathfrak{gl}(n-1|n-1) \supset \mathfrak{gl}(n-1|n-2) \supset \mathfrak{gl}(n-2|n-2) \supset 
\cdots \supset \mathfrak{gl}(1|1) \supset \mathfrak{gl}(1).
\label{chain}
\end{equation}
A particular feature of this odd GZ-basis, is that it can be extended~\cite{SV2016} for $n\rightarrow\infty$, which is not the case for the more conventional GZ-basis of $\gl(n|n)$~\cite{CGC,M,GIW1,GIW2}.
Combining the results of~\cite{SV2016} with the previous proposition, one has

\begin{prop}
For any positive integer $p$, a basis of the Fock representation $V(p,n)$ of $\B(n,n)$ is given by the set of vectors of the following form:
\beq
 |p;m)^{2n} \equiv |m)^{2n} = \left| \begin{array}{l} [m]^{2n} \\[2mm] |m)^{2n-1} \end{array} \right)= \hspace{5cm} 
\label{mn}
\eeq
\begin{equation*}
 \left|
\begin{array}{llcll:llclll}
m_{\bar n,2n} & m_{\overbar{n-1},2n} & \cdots & m_{\bar 2,2n} & m_{\bar 1,2n} & m_{1,2n} & m_{2,2n} &\cdots & m_{n-2,2n} &m_{n-1,2n} &m_{n,2n}\\
 \uparrow & \uparrow & \cdots & \uparrow &\uparrow & &&&&&\\
m_{\bar n, 2n-1} & m_{\overbar{n-1}, 2n-1} & \cdots & m_{\bar 2, 2n-1} & m_{\bar 1, 2n-1} & m_{1,2n-1} & m_{2,2n-1} &\cdots &m_{n-2,2n-1} & m_{n-1,2n-1} &\\
&&&&&\downarrow & \downarrow & \cdots & \downarrow &\downarrow  \\
 & m_{\overbar{n-1},2n-2} & \cdots & m_{\bar 2,2n-2} & m_{\bar 1,2n-2} & m_{1,2n-2} & m_{2,2n-2} &\cdots & m_{n-2,2n-2} & m_{n-1,2n-2} &\\
 &\uparrow & \cdots & \uparrow &\uparrow &&&&\\
 & m_{\overbar{n-1},2n-3} & \cdots & m_{\bar 2,2n-3} & m_{\bar 1,2n-3} & m_{1,2n-3} & m_{2,2n-3} &\cdots & m_{n-2,2n-3}  &  &\\
 &  &\ddots &\vdots & \vdots & \vdots &\vdots &\iddots & & \\
&&& m_{\bar 2 4} &  m_{\bar 1 4} & m_{14} & m_{24}& & & & \\
&&&\uparrow & \uparrow \\
&&& m_{\bar 2 3} &  m_{\bar 1 3} & m_{13} && & & & \\
&&&&&\downarrow \\
&&&& m_{\bar 1 2} & m_{12} & & & & & \\
&&&& \uparrow\\
&&&& m_{\bar 1 1}  & & & & & &
\end{array}
\right)
\end{equation*}
where all $m_{ij}\in\Z_+$, satisfying $m_{\bar{n},2n}\leq p$ and the GZ-conditions
\begin{equation}
 \begin{array}{rl}
1.& m_{j,2n}-m_{j+1,2n}\in{\mathbb Z}_+ , \;j\in[\bar{n},\bar{2}]\cup[1,n] \hbox{ and }
     m_{-1,2n}\geq \# \{i:m_{i,2n}>0,\; i\in[1,n] \};\\
2.& m_{-i,2s}-m_{-i,2s-1}\equiv\theta_{-i,2s-1}\in\{0,1\},\quad 1\leq i \leq s\leq n ;\\
3.& m_{i,2s}-m_{i,2s+1}\equiv\theta_{i,2s}\in\{0,1\},\quad 1\leq i\leq s\leq n-1 ;\\    
4.& m_{-1,2s}\geq \# \{i:m_{i,2s}>0,\; i\in[1,s] \}, \; s\in[1,n] ;\\  
5.& m_{-1,2s-1}\geq \# \{i:m_{i,2s-1}>0,\; i\in[1,s-1] \}, \; s\in[2,n] ;\\ 
6.& m_{i,2s}-m_{i,2s-1}\in{\mathbb Z}_+\hbox{ and } m_{i,2s-1}-m_{i+1,2s}\in{\mathbb Z}_+,\quad 1\leq i\leq s-1\leq n-1;\\
7.& m_{-i-1,2s+1}-m_{-i,2s}\in{\mathbb Z}_+\hbox{ and } m_{-i,2s}-m_{-i,2s+1}\in{\mathbb Z}_+,\quad 1\leq i\leq s\leq n-1.
 \end{array}
\label{cond3}
\end{equation}
\end{prop}
Conditions~2 and~3 are referred to as ``$\theta$-conditions''.
Conditions~6 and~7 are often referred to as ``betweenness conditions.''
Conditions~1, 4 and 5 are related to~\eqref{cond2} (or~\eqref{mpos}), and assure that each row of~\eqref{mn} corresponds to the highest weight of a covariant representation of $\gl(t|t)$ or $\gl(t|t-1)$ in the chain~\eqref{chain}.
Note that the arrows in this pattern have no real function, and can be omitted. 
We find it useful to include them, just in order to visualize the conditions.
When there is an arrow $a\rightarrow b$ between labels $a$ and $b$, it means that either $b=a$ or else $b=a+1$ (a $\theta$-condition).
We will also refer to ``rows'' and ``columns'' of the GZ-pattern.
Rows are counted from the bottom: row~1 is the bottom row in~\eqref{mn}, and row~$2n$ is the top row in~\eqref{mn}.
In an obvious way, columns~$1$, $2$, $3\cdots$ refer to the columns to the right of the dashed line in~\eqref{mn},
and columns~$-1$, $-2$, $-3,\cdots$ (or $\bar 1$, $\bar 2$, $\bar 3,\cdots$) to the columns to the left of this dashed line.
For two consecutive rows in the GZ-pattern~\eqref{mn}, about half of the labels involve $\theta$-conditions, and the other half involves betweenness conditions.

The top row of~\eqref{mn} corresponds to the highest weight of a $\gl(n|n)$ covariant representation, according to~\eqref{cond1}-\eqref{cond2}.
The action of a set of $\gl(n|n)$ generators $E_{ij}$ on such basis vectors of a covariant representation has been determined in~\cite{SV2016}. 
For future reference, let us recall the action of the elements of the Cartan subalgebra of $\gl(n|n)$ or $\B(n,n)$, which involves now the label $p$ since the $E_{ii}$'s have a nonzero action on the vacuum (i.e.\ the zero GZ-pattern):
\begin{align}
& E_{-i,-i}|m)^{2n}=\left(-\frac{p}{2}+\sum_{j\in[-i,i-1]^*} m_{j,2i-1}-\sum_{j\in[-i+1,i-1]^*} m_{j,2i-2}\right)|m)^{2n}, 
\quad i\in[1,n]; \label{E-ii}\\
& E_{ii}|m)^{2n}=\left(\frac{p}{2}+\sum_{j\in[-i,i]^*} m_{j,2i}-\sum_{j\in[-i,i-1]^*} m_{j,2i-1}\right)|m)^{2n}, 
\quad i\in[1,n]. \label{Eii}
\end{align}
Here, we want to go beyond the action of the $E_{ij}$'s, and describe the action of the generators $c_i^\pm$ of $\B(n,n)$, i.e.\ the parastatistics operators, on the basis vectors~\eqref{mn}.
For the parastatistics creation operators, one can check that
\beq
\lb E_{ij}, c^+_k \rb = \delta_{jk} c^+_i,
\eeq
so the set $\{c^+_i | i\in[\bar{n},n]^*\}$ forms a standard $\gl(n|n)$ tensor of rank $[1,0,\ldots,0]$ (which is the highest weight of the standard covariant representation).
This means that to every $c_i^+$ one can attach a unique GZ-pattern of the form~\eqref{mn} with top line $1 0\ldots 0$.
Following the $\gl(n|n)$ actions on GZ-patterns (in particular, \cite[(3.15)-(3.16)]{SV2016}), it is easy to see that the $2n$ elements $(c^+_n, c^+_{-n}, \cdots, c^+_2, c^+_{-2},c^+_1, c^+_{-1})$ correspond, in this order, to a GZ-pattern of type~\eqref{mn} consisting of $k$ top rows of the form $1 0 \cdots 0$ and $2n-k$ bottom rows of the form $0 \cdots 0$ for $k=1,2,\ldots,2n$. 
For example, for $n=3$:
\begin{align}
& c^+_3 : \left|
\begin{array}{ccc:ccc}
 1&0&0&0&0&0 \\
 0&0&0&0&0&  \\
  &0&0&0&0&  \\
  &0&0&0& &  \\
  & &0&0& &  \\
  & &0& & &  
\end{array}
\right),
\qquad
c^+_{-3} : \left|
\begin{array}{ccc:ccc}
 1&0&0&0&0&0 \\
 1&0&0&0&0&  \\
  &0&0&0&0&  \\
  &0&0&0& &  \\
  & &0&0& &  \\
  & &0& & &  
\end{array}
\right),
\qquad
c^+_{2} : \left|
\begin{array}{ccc:ccc}
 1&0&0&0&0&0 \\
 1&0&0&0&0&  \\
  &1&0&0&0&  \\
  &0&0&0& &  \\
  & &0&0& &  \\
  & &0& & &  
\end{array}
\right), \nn \\
& c^+_{-2} : \left|
\begin{array}{ccc:ccc}
 1&0&0&0&0&0 \\
 1&0&0&0&0&  \\
  &1&0&0&0&  \\
  &1&0&0& &  \\
  & &0&0& &  \\
  & &0& & &  
\end{array}
\right),
\qquad
c^+_{1} : \left|
\begin{array}{ccc:ccc}
 1&0&0&0&0&0 \\
 1&0&0&0&0&  \\
  &1&0&0&0&  \\
  &1&0&0& &  \\
  & &1&0& &  \\
  & &0& & &  
\end{array}
\right),
\qquad
c^+_{-1} : \left|
\begin{array}{ccc:ccc}
 1&0&0&0&0&0 \\
 1&0&0&0&0&  \\
  &1&0&0&0&  \\
  &1&0&0& &  \\
  & &1&0& &  \\
  & &1& & &  
\end{array}
\right).
\label{cGZ}
\end{align}
It will be useful to express this correspondence by means of a row function:
\beq
\rho(i)= \left\{ \begin{array}{rcl}
 {2i} &  \hbox{for} &  i\in[1,n]  \\ 
 {-2i-1} &  \hbox{for} &  i\in[\bar{n},\bar{1}]
 \end{array}\right. .
\label{rho}
\eeq
Then the pattern corresponding to $c^+_i$ has rows of the form $1 0 \cdots 0$ for each row index $j\in[\rho(i),2n]$ and zero rows for each row index $j\in[1,\rho(i)-1]$.  

The tensor product rule for covariant representations of $\gl(n|n)$ is well known~\cite{King1990}.
When the $\gl(n|n)$ representation with highest weight $[m]^{2n}$ is denoted by $W([m]^{2n})$, it reads
\begin{equation}
W([1,0,\ldots,0]) \otimes W([m]^{2n}) = \bigoplus_{k\in[-n,n]^*} W([m]_{+(k)}^{2n}), \label{tensprod}
\end{equation}
where in general $[m]_{\pm(k)}^{2n}$ is obtained from $[m]^{2n}$ by the replacement of $m_{k,2n}$ by $m_{k,2n}\pm 1$.
On the right hand side of~(\ref{tensprod}) the summands for which the conditions~(\ref{cond1})-(\ref{cond2}) are not fulfilled are omitted. 

By standard analysis~\cite{Vilenkin,paraboson}, the matrix elements of $c_i^+$ in $V(p,n)$ can be written as follows:
\begin{align}
{}^{2n}(m' | c_i^+ | m )^{2n} & = 
\left( \begin{array}{ll} [m]^{2n}_{+(k)} \\[1mm] |m')^{2n-1} \end{array} \right| c_i^+
\left| \begin{array}{ll} [m]^{2n} \\[1mm] |m)^{2n-1} \end{array} \right) \nn\\
& = \left( 
\begin{array}{c}1 0 \cdots 0 0\\[-1mm] 1 0 \cdots 0\\[-1mm] \cdots\\[-1mm] 0 \end{array} ;
\begin{array}{ll} [m]^{2n} \\[2mm] |m)^{2n-1} \end{array}  \right.
\left| 
\begin{array}{ll} [m]^{2n}_{+(k)} \\[2mm] |m')^{2n-1} \end{array} \right)
\times
([m]^{2n}_{+(k)}||c^+||[m]^{2n}).
\label{mmatrix}
\end{align}
The GZ-pattern with 0's and 1's is of course the one corresponding to $c^+_i$, as described earlier.
The first factor in the right hand side of~\eqref{mmatrix} is a $\gl(n|n)$ Clebsch-Gordan coefficient (CGC), where all patterns are of the form~\eqref{mn}.
These CGC's will be determined and given in Appendix~\ref{A}.
From the computational point of view, this is the main contribution of the current paper.
The second factor in~\eqref{mmatrix} is a {\em reduced matrix element} for the standard representation.
The possible values of the patterns $|m')^{2n}$ are determined by the $\gl(n|n)$ tensor product rule and the first line of $|m')^{2n}$
is of the form $[m]^{2n}_{+(k)}$.

It is important to realize that the reduced matrix elements depend only upon the $\gl(n|n)$ highest weights $[m]^{2n}$ and $[m]^{2n}_{+k}$ (and not on the type of GZ basis that is being used.)
These reduced matrix elements have actually been determined in~\cite[Proposition~4]{SV2015}:
\begin{align}
& ([m]^{2n}_{+(k)}||c^+||[m]^{2n}) = 
G_{n+k+1}(m_{-n,2n},m_{-n+1,2n},\ldots,m_{-1,2n},m_{1,2n},\ldots,m_{2n,2n}),\quad (k\in[-n,-1]) \\
& ([m]^{2n}_{+(k)}||c^+||[m]^{2n}) = 
G_{n+k}(m_{-n,2n},m_{-n+1,2n},\ldots,m_{-1,2n},m_{1,2n},\ldots,m_{2n,2n}),\quad (k\in[1,n]).
\end{align}
 
For the matrix elements of $c_i^-$, we use the Hermiticity requirement~\eqref{Fock}, 
\beq
{}^{2n}(m'|c^-_i|m)^{2n} = {}^{2n}(m|c^+_i|m')^{2n}.
\label{herm}
\eeq
So in this way we obtain explicit actions of the $\B(n,n)$ generators $c_i^\pm$ on a basis of $V(p,n)$:
\begin{align}
c_i^+|m)^{2n} & = \sum_{k,m'} 
\left( 
\begin{array}{c}1 0 \cdots 0 0\\[-1mm] 1 0 \cdots 0\\[-1mm] \cdots\\[-1mm] 0 \end{array} ;
\begin{array}{ll} [m]^{2n} \\[2mm] |m)^{2n-1} \end{array}  \right.
\left| 
\begin{array}{ll} [m]^{2n}_{+(k)} \\[2mm] |m')^{2n-1} \end{array} \right)
([m]^{2n}_{+(k)}||c^+||[m]^{2n}) 
\left| \begin{array}{ll} [m]^{2n}_{+(k)} \\[1mm] |m')^{2n-1} \end{array} \right), \label{cj+r}\\
c_i^-|m)^{2n} & = \sum_{k,m'} 
\left( 
\begin{array}{c}1 0 \cdots 0 0\\[-1mm] 1 0 \cdots 0\\[-1mm] \cdots\\[-1mm] 0 \end{array} ;
\begin{array}{ll} [m]^{2n}_{-(k)} \\[2mm] |m')^{2n-1} \end{array}  \right.
\left| 
\begin{array}{ll} [m]^{2n} \\[2mm] |m)^{2n-1} \end{array} \right)
([m]^{2n}||c^+||[m]^{2n}_{-(k)}) 
\left| \begin{array}{ll} [m]^{2n}_{-(k)} \\[1mm] |m')^{2n-1} \end{array} \right). \label{cj-r}
\end{align}

\section{Properties of the Fock representations $V(p,n)$}
\setcounter{equation}{0} \label{sec:D}

Before turning to the case where $n$ goes to infinity, it is useful to collect some properties of the actions of the parastatistics creation and annihilation operators $c_i^\pm$ on basis vectors $|m)^{2n}$ of $V(p,n)$.
For this purpose, let us write~\eqref{cj+r}-\eqref{cj-r} as
\begin{align}
& c_i^+|m)^{2n} = \sum_{m'} C^+\left[i,|m)^{2n},|m')^{2n}\right]\; |m')^{2n} \label{C+} ,\\
& c_i^-|m)^{2n} = \sum_{m'} C^-\left[i,|m)^{2n},|m')^{2n}\right]\; |m')^{2n} \label{C-} .
\end{align}
The coefficients $C^\pm\left[i,|m)^{2n},|m')^{2n}\right]$ are just a shorthand notation for the expressions in~\eqref{cj+r}-\eqref{cj-r}, and all parts of these expressions (CGC's and reduced matrix elements) are explicitly known.

One can think of the GZ-vectors of $V(p,n)$ as follows.
The vacuum vector $|0\rangle$ is the GZ basis vector with all zeros: $|0\rangle = |p;0)^{2n} \equiv |0)^{2n}$. 
For example, for $n=3$,
\beq
|0\rangle = \left|
\begin{array}{ccc:ccc}
 0&0&0&0&0&0 \\
 0&0&0&0&0&  \\
  &0&0&0&0&  \\
  &0&0&0& &  \\
  & &0&0& &  \\
  & &0& & &  
\end{array}
\right).
\label{00}
\eeq
The creation operators $c_i^+$ ($i\in[\bar{n},n]^*$) have the effect of increasing certain entries in the GZ-pattern, and the annihilation operators $c_i^-$ ($i\in[\bar{n},n]^*$) decrease certain entries in the GZ-pattern (which is why $|0)^{2n}$ is the vacuum vector).

Let us concentrate on the action of the creation operators.
From the properties of the CGC's (see Appendix), we have immediately the following:
\begin{prop}
Let $|m)^{2n}$ be a basis vector with a valid GZ-pattern (i.e.\ satisfying~\eqref{cond3}).
Then the only patterns $|m')^{2n}$ appearing in the right hand side of 
\[
 c_i^+|m)^{2n} = \sum_{m'} C^+\left[i,|m)^{2n},|m')^{2n}\right]\; |m')^{2n} \quad (i\in[\bar{n},n]^*)
\]
are valid GZ-patterns $|m')^{2n}$ such that
\begin{align}
& [m']^j = [m]^j \hbox{ for } j\in[1,\rho(i)-1], \nn\\
& [m']^j = [m]^j+[0,\ldots,0,1,0,\ldots,0] \hbox{ for } j\in[\rho(i),2n].
\label{changes}
\end{align}
\label{prop5}
\end{prop}
In other words, there are no changes in the entries of rows $1,2,\ldots,\rho(i)-1$. 
And in rows $\rho(i),\ldots,2n$ there is a change by one unit for just one particular column index $s$.
The increase can be in any possible column, as long as the remaining pattern is still valid, i.e.\ as long as~\eqref{cond3} is satisfied.
For example, for $n=3$ the action of $c_i^+$ on $|0\rangle = |0)^{6}$ will give just one vector, namely the corresponding one appearing in~\eqref{cGZ}.
Then, e.g., $c^+_{-3}c^+_{-2} |0)^{6}$ will give a linear combination of only two vectors, namely
\beq
\left|
\begin{array}{ccc:ccc}
 2&0&0&0&0&0 \\
 2&0&0&0&0&  \\
  &1&0&0&0&  \\
  &1&0&0& &  \\
  & &0&0& &  \\
  & &0& & &  
\end{array}
\right)
\quad \hbox{and} \quad
\left|
\begin{array}{ccc:ccc}
 1&1&0&0&0&0 \\
 1&1&0&0&0&  \\
  &1&0&0&0&  \\
  &1&0&0& &  \\
  & &0&0& &  \\
  & &0& & &  
\end{array}
\right),
\label{example1}
\eeq
because all other ways of adding 1's to row 5 and 6 (by the action of $c^+_{-3}$) of the pattern
\beq
c^+_{-2} |0\rangle = \sqrt{p} \left|
\begin{array}{ccc:ccc}
 1&0&0&0&0&0 \\
 1&0&0&0&0&  \\
  &1&0&0&0&  \\
  &1&0&0& &  \\
  & &0&0& &  \\
  & &0& & &  
\end{array}
\right)
\label{example2}
\eeq
give rise to nonvalid GZ-patterns.
Note that we have implicitly assumed that $p\geq 2$, since for $p=1$ the first vector in~\eqref{example1} is zero.

For each row $k$ of a GZ basis vector $|m)^{2n}$, let $|m_k|$ be the sum of all entries in row $k$, i.e.
\beq
|m_{2k}|=\sum_{i\in[-k,k]^*} m_{i,2k}, \qquad |m_{2k+1}|=\sum_{i\in[-k-1,k]^*} m_{i,2k+1}.
\label{m-wt}
\eeq
We will sometimes refer to $|m_k|$ as the weight of row~$k$.
By~\eqref{changes}, one has:
\beq
|m_{1}|\leq |m_{2}|\leq |m_{3}|\leq \cdots \leq |m_{2n-1}|\leq |m_{2n}|.
\label{ineq}
\eeq

Consider row~$s$ for a general GZ pattern $|m)^{2n}$:
\[
[\ldots,m_{\bar{2},s}, m_{\bar{1},s} ; m_{1,s}, m_{2,s},\ldots].
\]
By the betweenness conditions in~\eqref{cond3}, it follows that both parts $(\ldots,m_{\bar{2},s}, m_{\bar{1},s})$ and
$(m_{1,s}, m_{2,s},\ldots)$ are partitions~\cite{Macdonald}, since they consist of non-decreasing non-negative integers.
We shall refer to these partitions as the left and right part of row~$s$.
Furthermore, if $m_{\bar{1},s}=0$, it follows from conditions 4 and 5 in~\eqref{cond3} that the right part of row $s$ must be zero, i.e.\ in that case row $s$ is of the form
\beq
[\ldots,m_{\bar{3},s}, m_{\bar{2},s},0 ; 0,0,\ldots]=[\nu_1,\nu_2,\ldots,0;0,0,\ldots],
\eeq
with $\nu$ some partition. 
The following definition will be crucial:
\begin{defi}
The pattern, or equivalently the associated basis vector, $|m)^{2n}$ is row-stable with respect to row $s$ if there exists a partition $\nu$ such that all rows $s,s+1,\ldots,2n$ are of the form
\[
[\nu_1,\nu_2,\ldots,0;0,0,\ldots].
\]
In that case, $s$ is called a stability index of $|m)^{2n}$.
\label{defi6}
\end{defi}
Note that the length of these rows increases as the row index increases, but the increase is only by adding extra zeros in the left and right part of the row.
For example, \eqref{example2} is row-stable with respect to row~3, and the patterns in~\eqref{example1} are row-stable with respect to row~5.

Consider now the consecutive action of a number of $c^+_i$'s, and suppose $n$ is sufficiently large. We have
\begin{prop}
Let $k<n$, then all basis vectors appearing in 
\beq
c^+_{i_k} \cdots c^+_{i_2} c^+_{i_1} |0\rangle \qquad(\hbox{each } i_r\in[\bar{n},n]^*)
\label{consec-action}
\eeq
are row-stable with respect to some row index $s$.
\label{prop7}
\end{prop}

\begin{proof}
Starting with the zero pattern, the action of each $c^+_{i_r}$ adds a $+1$ to the obtained pattern in some position of row~$j$, for $j\in[\rho(i),2n]$. 
This happens in such a way that the left parts of all the rows are partitions.
Suppose we are at the end of the action~\eqref{consec-action}, and that at that point there is some row index $s$ where the partition in the left part is of the form $[\nu_1,\nu_2,\ldots,0]$, with $|\nu|=k$.
Then by~\eqref{ineq} all rows above $s$ have the same weight, 
and by conditions~3 and~6 of~\eqref{cond3} this means that all rows $j$ with $j>s$ are of the form $[\nu_1,\nu_2,\ldots,0;0,0,\ldots]$.
Thus the pattern obtained is row-stable with respect to row $s$.
Since the length of $\nu$ satisfies $\ell(\nu)\leq k<n$, there is at least one row $s\leq 2n$ such that the left part is of the form $[\nu_1,\nu_2,\ldots,0]$, i.e.\ ending with a 0.
(In the ``worst case'', row~$2n$ is of the form $[1,1,\ldots,0;0,\ldots,0]$.)
\end{proof}

Row-stable patterns are in some sense preserved under the action of $c^+_i$'s.

\begin{prop}
Let $|m)^{2n}$ be row-stable with respect to row $s$, where $s<2n-1$.
Then the vectors $|m')^{2n}$ appearing in $c^+_i |m)^{2n}$ are row-stable with respect to row $\max\{s+2,\rho(i)+1\}$.
\label{prop8}
\end{prop}

\begin{proof}
There are two cases to be considered: $\rho(i)\leq s$ and $\rho(i)>s$.\\
1. Consider first the action of $c^+_i$ with $\rho(i)\leq s$.
Then 1's are added in all rows $j$, $j\in[\rho(i),2n]$, including row $s$.
There are two subcases to be analyzed: $s$ even and $s$ odd.\\
1.1. Let $s$ be odd. 
As a generic example, let row $s$ of $|m)^{2n}$ be given by $[3,2,1,0;0,0,0]$. 
After the action of $c^+_i$, row $s$ in the resulting vectors $|m')^{2n}$ is one of the following:
\[
(1)  \  [4,2,1,0;0,0,0],\quad
(2)  \  [3,3,1,0;0,0,0],\quad
(3)  \  [3,2,2,0;0,0,0],\quad
(4)  \  [3,2,1,1;0,0,0].
\]
By condition~3 of~\eqref{cond3}, row $s+1$ is respectively of the form
\[
(1)  \ [4,2,1,0;0,0,0,0],\quad
(2)  \ [3,3,1,0;0,0,0,0],\quad
(3)  \ [3,2,2,0;0,0,0,0],\quad
(4)  \ [3,2,1,1;0,0,0,0].
\]
It is clear that for (1), (2) and (3) the resulting vector is row-stable with respect to row $s$.
For (4), the conditions~\eqref{cond3} imply that row $s+2$ must be equal to $[3,2,1,1,0;0,0,0,0]$, and then the vector is row-stable with respect to row $s+2$.
Clearly, this argument works whenever row $s$ is of the form $[\nu_1,\nu_2,\ldots,0;0,\ldots,0]$.\\
1.2. Let $s$ be even. 
We first work again with a generic example, say row $s$ of $|m)^{2n}$ is given by $[4,3,1,0;0,0,0,0]$. 
After the action of $c^+_i$, row $s$ in the resulting vectors $|m')^{2n}$ is one of the following:
\[
(1)  \ [5,3,1,0;0,0,0,0],\quad
(2)  \ [4,4,1,0;0,0,0,0],\quad
(3)  \ [4,3,2,0;0,0,0,0],\quad
(4)  \ [4,3,1,1;0,0,0,0].
\]
Again by conditions~\eqref{cond3}, row $s+1$ is then respectively
\[
(1)  \ [5,3,1,0,0;0,0,0,0],\quad
(2)  \ [4,4,1,0,0;0,0,0,0],\quad
(3)  \ [4,3,2,0,0;0,0,0,0],\quad
(4)  \ [4,2,1,1,0;0,0,0,0].
\]
For (1), (2) and (3) the resulting vector is row-stable with respect to row $s$.
For (4), the resulting vector is row-stable with respect to row $s+1$.
The argument generalizes when row $s$ is of the form $[\nu_1,\nu_2,\ldots,0;0,\ldots,0]$.\\
2. Consider next the action of $c^+_i$ with $\rho(i)> s$.
It is now a matter of considering the entries in row $\rho(i)$ of $|m)^{2n}$.
Exactly the same arguments as in cases 1.1 and 1.2 work here (with $\rho(i)$ taking over the role of $s$), except that situation (4) of 1.1 does not appear (since $\rho(i)>s$).
Hence in this case $|m')^{2n}$ is row-stable with respect to row $\rho(i)$ or $\rho(i)+1$.
\end{proof}

For the action of the annihilation operators $c^-_i$, the analysis is similar, but simpler.
Similarly as in Proposition~\ref{prop5}, when 
\[
 c_i^-|m)^{2n} = \sum_{m'} C^-\left[i,|m)^{2n},|m')^{2n}\right]\; |m')^{2n} \quad (i\in[\bar{n},n]^*),
\]
then the only patterns $|m')^{2n}$ appearing in the right hand side of the above expression\footnote 
{To clarify the meaning, a basis vector is said to appear in an expression (where the expression is written as a linear combination of basis vectors) if its coefficient is nonzero in this linear combination.}
are valid GZ-patterns $|m')^{2n}$ such that
\begin{align}
& [m']^j = [m]^j \hbox{ for } j\in[1,\rho(i)-1], \nn\\
& [m']^j = [m]^j+[0,\ldots,0,-1,0,\ldots,0] \hbox{ for } j\in[\rho(i),2n].
\label{changes-}
\end{align}

Suppose now that $|m)^{2n}$ has stability index $s$. Then, in particular,
\beq
|m_{1}|\leq |m_{2}|\leq \cdots \leq |m_{s}| =|m_{s+1}| =\cdots = |m_{2n-1}|= |m_{2n}|.
\eeq
If one acts on $|m)^{2n}$ by $c^-_{i}$ with $\rho(i)>s$, then $|m_s|$ would remain unchanged, whereas all $|m_j|$ with $j\in[\rho(i),2n]$ are decreased by 1. 
That would violate~\eqref{ineq}.
Hence the action of $c^-_{i}$ with $\rho(i)>s$ on a vector $|m)^{2n}$ with stability index $s$ is zero.

If one acts on $|m)^{2n}$ by $c^-_{i}$ with $\rho(i)\leq s$, then the combinatorics of the conditions~\eqref{cond3} imply that in the resulting vectors all partitions must be the same in (the left parts of) rows $s,s+1,\ldots,2n$.
Hence the action of $c^-_{i}$ with $\rho(i)\leq s$ on a vector $|m)^{2n}$ with stability index $s$ can only yield vectors $|m')^{2n}$ which have $s$ as stability index. Concluding:

\begin{prop}
Let $|m)^{2n}$ be row-stable with respect to row $s$.
Then the vectors $|m')^{2n}$ appearing in $c^-_i |m)^{2n}$ are row-stable with respect to row $s$.
\label{prop9}
\end{prop}

Note that Proposition~\ref{prop7}, where a string of creation operators was used, can be extended for a string of creation and/or annihilation operators. 
Consider
\beq
c_{i_k}^{\eta_k} \cdots c_{i_2}^{\eta_2} c_{i_1}^{\eta_1} |0\rangle, \qquad(\eta_r\in\{+1,-1\},\  i_r\in[\bar{n},n]^*),
\label{string}
\eeq
with $0\leq |\eta|\equiv\eta_1+\eta_2+\cdots +\eta_k<n$.
Then all basis vectors appearing in~\eqref{string} are row-stable with respect to some row~$s$.
Note that in this case, the partition appearing in the stable row,
\[
[\nu_1,\nu_2,\ldots,0;0,\ldots,0],
\]
has $|\nu|=|\eta|$.

To complete this section, we will also establish a stability property of the coefficients in~\eqref{C+}-\eqref{C-}.
Suppose that the top row of $|m)^{2n}$ has the zero partition as second part, i.e.\ it is of the form
\[
[m]^{2n}=[\nu_1,\nu_2,\ldots;0,\ldots,0]
\]
with $\nu$ a partition.
Define a map $\phi_{2n,+2}$ from the set of GZ-patterns $|m)^{2n}$ with zero second part to the set of GZ-patterns $|m)^{2n+2}$ with stability index $2n$ by:
\begin{align}
& |m)^{2n+2}=\phi_{2n,+2}\left( |m)^{2n} \right),\label{phi}\\
& \hbox{with }
[m]^{2n+1}=[\nu_1,\nu_2,\ldots,0,0;0,\ldots,0],\quad [m]^{2n+2}=[\nu_1,\nu_2,\ldots,0,0;0,\ldots,0,0]. \nn
\end{align}
In other words, the top row of $|m)^{2n}$ is just repeated twice, with the extra addition of zeros in order to have sufficient entries for the pattern $|m)^{2n+2}$.
Clearly, the action of $\phi_{2n,+2}$ can also be extended by linearity, on a linear combination of vectors $|m)^{2n}$ with zero second part.

\begin{prop}
\label{prop10}
Let $|m)^{2n}$ be row-stable with respect to row~$2n$, and $|m)^{2n+2}=\phi_{2n,+2}\left( |m)^{2n} \right)$
Then for all $i$ with $\rho(i)\leq 2n$ (or equivalently, $i\in[-n,n]^*$):
\[
c^+_i |m)^{2n+2} = \phi_{2n,+2}\left(c^+_i|m)^{2n}\right).
\]
In other words, if 
\beq
c_i^+|m)^{2n} = \sum_{m'} C^+\left[i,|m)^{2n},|m')^{2n}\right]\; |m')^{2n} 
\label{eq10}
\eeq
then
\[
c_i^+|m)^{2n+2} = \sum_{m'} C^+\left[i,|m)^{2n},|m')^{2n}\right]\; \phi_{2n,+2}\left(|m')^{2n}\right); 
\]
or otherwise said:
\[
C^+\left[i,|m)^{2n+2},|m')^{2n+2}\right] = C^+\left[i,|m)^{2n},|m')^{2n}\right].
\]
\end{prop}

\begin{proof}
First of all, all vectors $|m')^{2n}$ in~\eqref{eq10} are such that the 2nd part of row $2n$ consists of zeros, so $\phi_{2n,+2}\left( |m')^{2n} \right)$ is well defined.
Consider now $C^+[i,|m)^{2n+2},|m')^{2n+2}]$, of the form
\beq
\left( 
\begin{array}{c}1 0 \cdots 0 0\\[-1mm] 1 0 \cdots 0\\[-1mm] \cdots\\[-1mm] 0 \end{array} ;
\begin{array}{ll} [m]^{2n+2} \\[2mm] |m)^{2n+1} \end{array}  \right.
\left| 
\begin{array}{ll} [m]^{2n+2}_{+(k)} \\[2mm] |m')^{2n+1} \end{array} \right)
([m]^{2n+2}_{+(k)}||c^+||[m]^{2n+2}) 
\eeq
for some row index $k$ with $k\in[-n-1,-2]$.
It is now a matter of inspecting the CGC and reduced matrix element in this case.
For the reduced matrix element, $[m]^{2n+2}$ is of the form
\[
[\nu_1,\nu_2,\ldots,0,0;0,\ldots,0,0]
\]
and $[m]^{2n+2}_{+(k)}$ is of the same form, with one of the parts of $\nu$ increased by 1.
From the explicit form of these reduced matrix elements~\cite[(3.24)-(3.25)]{SV2015}, it is not difficult to see that
\[
([m]^{2n+2}_{+(k)}||c^+||[m]^{2n+2}) = ([m]^{2n}_{+(k)}||c^+||[m]^{2n}).
\]
Following the notation of the Appendix, the CGC of $\gl(n+1|n+1)$ can be written as a product of two isoscalar factors times a CGC for $\gl(n|n)$:
\begin{align}
& \left( \begin{array}{c}10\cdots 00\\10\cdots 0\\\cdots \\ 0  \end{array};
\begin{array}{ll} [m]^{2n+2} \\[2mm] |m)^{2n+1} \end{array}  \right.
 \left| \begin{array}{ll} [m]^{2n+2}_{+(k)} \\[2mm] |m^\prime)^{2n+1} \end{array} \right)
 =  
\left( \begin{array}{l} 1\dot{0} \\ 1 \dot{0} \end{array}
\left| \begin{array}{l}  [m]^{2n+2} \\ {[m]}^{2n+1} \end{array} \right.\right|
\left. \begin{array}{l} [m]^{2n+2}_{+(k)}  \\ {[m^\prime]}^{2n+1} \end{array} \right)\times \nonumber\\
& \times\left( \begin{array}{l} 1\dot{0} \\ 1 \dot{0} \end{array} 
\left| \begin{array}{l}  [m]^{2n+1} \\ {[m]}^{2n} \end{array} \right.\right|
\left. \begin{array}{l} [m^\prime]^{2n+1}  \\ {[m^\prime]}^{2n} \end{array} \right) 
\left( \begin{array}{c} 10\cdots 00\\ \cdots \\ 0\end{array};
\begin{array}{ll} [m]^{2n} \\[2mm] |m)^{2n-1} \end{array}  \right.
 \left| \begin{array}{ll} [m]^{2n}_{+(k+1)} \\[2mm] |m^\prime)^{2n-1} \end{array} \right).
\label{red}
\end{align} 
By construction, up to additional 0's at the end of both parts, $[m]^{2n+2}_{+(k)}$, ${[m^\prime]}^{2n+1}$ and ${[m^\prime]}^{2n}=[m]^{2n}_{+(k+1)}$ are identical, all of the form
\[
[\mu_1,\mu_2,\ldots,0;0,\ldots,0]
\]
where $\mu$ is a partition obtained from $\nu$ by increasing one part by 1.
But for such special values, the isoscalar factors (given in the Appendix) simplify:
\[
\left( \begin{array}{l} 1\dot{0} \\ 1 \dot{0} \end{array}
\left| \begin{array}{l}  [\nu_1,\nu_2,\ldots,0;0\ldots,0]^{2n+2} \\ {[\nu_1,\nu_2,\ldots,0;0\ldots,0]}^{2n+1} \end{array} \right.\right|
\left. \begin{array}{l} [\mu_1,\mu_2,\ldots,0;0\ldots,0]^{2n+2}  \\ {[\mu_1,\mu_2,\ldots,0;0\ldots,0]}^{2n+1} \end{array} \right) = 1
\]
by~\eqref{liso3},
and
\[
\left( \begin{array}{l} 1\dot{0} \\ 1 \dot{0} \end{array}
\left| \begin{array}{l}  [\nu_1,\nu_2,\ldots,0;0\ldots,0]^{2n+1} \\ {[\nu_1,\nu_2,\ldots,0;0\ldots,0]}^{2n} \end{array} \right.\right|
\left. \begin{array}{l} [\mu_1,\mu_2,\ldots,0;0\ldots,0]^{2n+1}  \\ {[\mu_1,\mu_2,\ldots,0;0\ldots,0]}^{2n} \end{array} \right) = 1
\]
by~\eqref{Sliso3}.
Thus we obtain the result.
\end{proof}
Note that Proposition~\ref{prop10} is also valid if one replaces $c^+_i$ by $c^-_i$.

\section{The Lie superalgebra $\B(\infty,\infty)$ and its Fock representations $V(p)$}
\setcounter{equation}{0} \label{sec:E}

We are now in a position that we can extend both the parastatistics algebra $\B(n,n)$ and its Fock representations $V(p,n)$ to the infinite rank case $\B(\infty,\infty)$.
As usual for infinite rank Lie algebras or Lie superalgebras, the matrix form will consist of certain infinite matrices with a finite number of non-zero elements~\cite{Kac-Raina,Kac3,Penkov}.

Consider the set of all squared infinite matrices of the form
\beq
X:=
\left( \begin{array}{ccc:c|ccc}
\ddots & \vdots & \vdots & \vdots & \vdots & \vdots & \iddots \\
\cdots & X_{\bar{2},\bar{2}} & X_{\bar{2},\bar{1}} & X_{\bar{2},0} & X_{\bar{2},1} & X_{\bar{2},2} & \cdots\\
\cdots & X_{\bar{1},\bar{2}} & X_{\bar{1},\bar{1}} & X_{\bar{1},0} & X_{\bar{1},1} & X_{\bar{1},2} & \cdots\\
\hdashline
\cdots & X_{{0},\bar{2}}  & X_{{0},\bar{1}} & 0 & X_{{0},1} & X_{{0},2}& \cdots \\
\hline
\cdots & X_{{1},\bar{2}} & X_{{1},\bar{1}} & X_{{1},0} & X_{{1},1} & X_{{1},2}& \cdots \\
\cdots & X_{{2},\bar{2}} & X_{{2},\bar{1}} & X_{{2},0} & X_{{2},1} & X_{{2},2}& \cdots \\
\iddots & \dots & \vdots & \vdots & \vdots & \vdots & \ddots 
\end{array} \right).
\label{Xi}
\eeq
where the above indices take values in the set $\Z$.
Herein, any matrix of the form $X_{ij}$ with $i,j \in\Z^*$ is a $(2\times 2)$-matrix, $X_{0,i}$ is a $(1\times 2)$-matrix and $X_{i,0}$ a $(2\times 1)$-matrix.
The infinite-dimensional Lie superalgebra $\B(\infty,\infty)$ can be defined as the set of all squared infinite matrices of the form~\eqref{Xi} such that each matrix has only a finite number of nonzero entries, and such that the (non-zero) blocks satisfy
\begin{align*}
& I X_{\bar{i},\bar{j}}+X_{\bar{j},\bar{i}}^T I=0, \quad
JX_{{i},{j}}+X_{{j},{i}}^T J=0, \quad
I X_{\bar{i},{j}}-X_{{j},\bar{i}}^T J=0 \qquad(i,j\in\Z_+^*);\\
& X_{0,\bar{j}}+X_{\bar{j},0}^T I=0, \quad X_{0,{j}}-X_{{j},0}^T  J=0 \qquad(j\in\Z_+).
\end{align*}
Such an element $X$ is even if $X_{ij}=0$ for all $(i,j)\in\Z_-\times Z_+^*$ and $(i,j)\in\Z_+^*\times\Z_-$, and $X$ is odd if  $X_{ij}=0$ for all $(i,j)\in\Z_-\times\Z_-$ and $(i,j)\in\Z_+^*\times\Z_+^*$.
For homogeneous elements $X$ and $Y$, the Lie superalgebra bracket is defined as usual,
\[
\lb X,Y\rb = XY-(-1)^{\deg(X)\deg(Y)}YX,
\]
and extended by linearity.

We can again consider the matrices $e_{ij}$ consisting of zeros everywhere except a $1$ on position $(i,j)$, where the row and column indices belong to $\Z$. 
A basis of a Cartan subalgebra $\h$ of $\B(\infty,\infty)$ consists of the elements $h_i=e_{2i-1,2i-1}-e_{2i,2i}$ ($i\in\Z_+^*$) and
$h_{i}=e_{2i,2i}-e_{2i+1,2i+1}$ ($i\in\Z_-^*$). 
The corresponding dual basis of $\h^*$ is denoted by $\epsilon_i$ ($i\in\Z^*$).
As in the finite rank case, we can identify the following even root vectors with roots $\epsilon_{-i}$ and $-\epsilon_{-i}$ respectively ($i\in\Z_+^*$):
\begin{align}
&c_{-i}^+\equiv f_{-i}^+= \sqrt{2}(e_{-2i,0}-e_{0,-2i+1}), \nn\\
&c_{-i}^-\equiv f_{-i}^-= \sqrt{2}(e_{0,-2i}-e_{-2i+1,0}), \;
\label{f-as-ei}
\end{align}
and odd root vectors with roots $\epsilon_{i}$ and $-\epsilon_{i}$ respectively ($i\in\Z_+^*$):
\begin{align}
&c_{i}^+ \equiv b_{i}^+= \sqrt{2}(e_{0,2i}+e_{2i-1,0}), \nn\\
&c_{i}^-\equiv b_{i}^-= \sqrt{2}(e_{0,2i-1}-e_{2i,0}). \; 
\label{b-as-ei}
\end{align}
The operators $c_i^+$ can be chosen as positive root vectors, and the $c_i^-$ as negative root vectors.

Just as in the finite rank case, the operators introduced here satisfy the triple relations of parastatistics.
So we are dealing with an infinite number of parafermions and an infinite number of parabosons, satisfying the mutual relative parafermion relations.
This is a straightforward extension of~\eqref{f-rel}-\eqref{fb-rel}, and these three sets of relations can be combined in a somewhat complicated form:
\begin{align}
& \lb\lb c_{ j}^{\xi}, c_{ k}^{\eta}\rb , c_{l}^{\epsilon}\rb =-2
\delta_{jl}\delta_{\epsilon, -\xi}\epsilon^{\langle l \rangle} 
(-1)^{\langle k \rangle \langle l \rangle }
c_{k}^{\eta} +2 \epsilon^{\langle l \rangle }
\delta_{kl}\delta_{\epsilon, -\eta}
c_{j}^{\xi}, \label{paraospi}\\
& \qquad\qquad \xi, \eta, \epsilon =\pm\hbox{ or }\pm 1;\quad j,k,l\in\Z^*, \nn 
\end{align} 
and we used the abbreviation $\langle i \rangle=\deg(c^\pm_i)$.

It is not difficult to see that Theorem~\ref{theo1} extends to the infinite rank case:
\begin{theo}
\label{theo11}
As a Lie superalgebra defined by generators and relations, 
$\B(\infty,\infty)$ is generated by the elements $c_i^\pm$ ($i\in\Z^*$) subject to the relations~\eqref{paraospi}.
\end{theo}

The parastatistics Fock space of order $p$, with $p$ a positive integer, can be defined as before, and will correspond to a lowest weight representation $V(p)$ of the algebra $\B(\infty,\infty)$.  
$V(p)$ is the Hilbert space generated by a vacuum vector $|0\rangle$ and the parastatistics creation and annihilation operators, i.e.\ subject to 
\begin{align}
& \langle 0|0\rangle=1, \qquad c_{j}^- |0\rangle = 0, 
\qquad (c_j^\pm)^\dagger = c_j^\mp, \qquad
\lb c_j^-, c_k^+ \rb |0\rangle = p\delta_{jk}\, |0\rangle \quad
(j,k\in\Z^*)
\label{Focki}
\end{align}
and which is irreducible  under the action of the algebra $\B(\infty,\infty)$.
Clearly $|0\rangle$ is a lowest weight vector of $V(p)$ with weight $[\ldots,-\frac{p}{2}, -\frac{p}{2}; \frac{p}{2}, \frac{p}{2}, \ldots]$ in the basis $\{\ldots,\epsilon_{-2}, \epsilon_{-1}, \epsilon_{1}, \epsilon_{2}, \ldots\}$.

In the following, we shall first describe the set of basis vectors of $V(p)$, then give the action of the operators $c^\pm_i$ on these basis vectors, and finally prove that under this action one is indeed dealing with an irreducible representation of the Lie superalgebra $\B(\infty,\infty)$.

The basic idea is that the GZ-patterns~\eqref{mn} consisting of $2n$ rows in the finite rank case can be extended to GZ-patterns with an infinite number of rows.
But not all infinite GZ-patterns are valid, only the ones that are ``row-stable'' will correspond to vectors of $V(p)$.
(Note that in a previous paper~\cite{SV2016}, we were also dealing with ``stable'' GZ-patterns of $\gl(\infty|\infty)$ in an odd basis. 
The ones appearing in that earlier paper could be called ``column-stable'', and are not the ones appearing here.)

Before giving a definition of infinite row-stable GZ-patterns, let us give an example.
Consider the pattern with an infinite number of rows:
\begin{equation}
|m)^{\infty} = \left|
\begin{array}{ccccc:ccccc}
 \ddots&\ddots& & & & & &\iddots&\iddots \\
  &4&3&1&0&0&0&0&0&  \\
  &4&3&1&0&0&0&0& &  \\
  & &4&3&1&0&0&0& &  \\
  & &4&3&1&0&0& & &  \\
  & & &3&3&1&0& & &  \\
  & & &3&2&1& & & &  \\
  & & & &2&2& & & &  \\
  & & & &1& & & & &  
\end{array}
\right).
\end{equation}
The definition of stability is essentially the same as in the finite case: the pattern $|m)^{\infty}$ is row-stable with respect to row $s$ if there exists a partition $\nu$ such that all rows $s,s+1,s+2,\ldots$ are of the form
\[
[\nu_1,\nu_2,\ldots,0;0,0,\ldots].
\]
In that case, $s$ is called a stability index of $|m)^{\infty}$.
In the above example, the smallest stability index is 7.

The following proposition describes the basis of $V(p)$.
\begin{prop}
A basis of $V(p)$ is given by all infinite row-stable GZ-patterns $|m)^{\infty}$ of the following form:
\beq
 |p;m)^{\infty} \equiv |m)^{\infty} =  \hspace{6cm} 
\label{mni}
\eeq
\begin{equation*}
 \left|
\begin{array}{llcll:llclll}
\ddots & \vdots &  & \vdots & \vdots & \vdots & \vdots &\cdots & \vdots & \iddots &\\
 & m_{\bar{n},2n} & \cdots & m_{\bar 2,2n} & m_{\bar 1,2n} & m_{1,2n} & m_{2,2n} &\cdots & m_{n-1,2n} & m_{n,2n} &\\
 &\uparrow & \cdots & \uparrow &\uparrow &&&&\\
 & m_{\overbar{n},2n-1} & \cdots & m_{\bar 2,2n-1} & m_{\bar 1,2n-1} & m_{1,2n-1} & m_{2,2n-1} &\cdots & m_{n-1,2n-1}  &  &\\
 &  &\ddots &\vdots & \vdots & \vdots &\vdots &\iddots & & \\
&&& m_{\bar 2 4} &  m_{\bar 1 4} & m_{14} & m_{24}& & & & \\
&&&\uparrow & \uparrow \\
&&& m_{\bar 2 3} &  m_{\bar 1 3} & m_{13} && & & & \\
&&&&&\downarrow \\
&&&& m_{\bar 1 2} & m_{12} & & & & & \\
&&&& \uparrow\\
&&&& m_{\bar 1 1}  & & & & & &
\end{array}
\right) =
\left|
\begin{array}{l} 
\vdots \\ {}[m]^{2n} \\ \phantom{\uparrow} \\ {}[m]^{2n-1} \\  \vdots \\ {}[m]^4 \\ \phantom{\uparrow}\\
{}[m]^3 \\ \phantom{\uparrow}\\{}[m]^2 \\ \phantom{\uparrow} \\ {}[m]^1 
\end{array}
\right)
\end{equation*}
where for each $|m)^\infty$ there should exist a row index $s$ (depending on $|m)^\infty$) such that row $s$ is of the form
\[
[m]^s=[\nu_1,\nu_2,\ldots,0;0,0,\ldots]
\]
with $\nu$ a partition, all rows above $s$ are of the same form (up to extra zeros), and $\nu_1\leq p$.
Furthermore all $m_{ij}\in\Z_+$ and the usual GZ-conditions should be satisfied:
\begin{equation}
 \begin{array}{rl}
1.& m_{-i,2r}-m_{-i,2r-1}\equiv\theta_{-i,2r-1}\in\{0,1\},\quad 1\leq i \leq r;\\
2.& m_{i,2r}-m_{i,2r+1}\equiv\theta_{i,2r}\in\{0,1\},\quad 1\leq i\leq r ;\\    
3.& m_{-1,2r}\geq \# \{i:m_{i,2r}>0,\; i\in[1,r] \}, \; r\in\Z_+^* ;\\  
4.& m_{-1,2r+1}\geq \# \{i:m_{i,2r+1}>0,\; i\in[1,r] \}, \; r\in\Z_+^* ;\\ 
5.& m_{i,2r+2}-m_{i,2r+1}\in{\mathbb Z}_+\hbox{ and } m_{i,2r+1}-m_{i+1,2r+2}\in{\mathbb Z}_+,\quad 1\leq i\leq r;\\
6.& m_{-i-1,2r+1}-m_{-i,2r}\in{\mathbb Z}_+\hbox{ and } m_{-i,2r}-m_{-i,2r+1}\in{\mathbb Z}_+,\quad 1\leq i\leq r.
 \end{array}
\label{cond3i}
\end{equation}
\end{prop}

In order to define the action of the generators $c^\pm_i$ on such vectors $|m)^\infty$, the map $\phi_{2n,+2}$ should be extended in an obvious way.
Let $|m)^{2n}$ be a finite GZ-pattern of type~\eqref{mn} with $2n$ rows, such that row $2n$ is of the form
\[
[\nu_1,\nu_2,\ldots;0,0,\ldots,0].
\]
Then $\phi_{2n,\infty}\left( |m)^{2n}\right)$ is the infinite GZ-pattern of type~\eqref{mni} consisting of the rows of $|m)^{2n}$ to which an infinite number of rows $[\nu_1,\nu_2,\ldots;0,0,\ldots,0]$ are added at the top (all identical, up to additional zeros).
Conversely, if an infinite GZ-pattern $|m)^{\infty}$ is given, which is stable with respect to row $2s$, then one can restrict the infinite pattern to a finite GZ-pattern, and 
\[
|m)^{2s} = \phi_{2s,\infty}^{-1}\left( |m)^{\infty}\right).
\]
Both maps can be extended by linearity.
Now we can define the action of $c^\pm_i$ on vectors $|m)^\infty$ of type~\eqref{mni}.
\begin{defi}
Given a vector $|m)^\infty$ of $V(p)$ with stability index $2s$, and a generator $c^\pm_i$.
Let $2n$ be such that $2n>\max\{2s,\rho(i)\}$.
Then 
\beq
c^\pm_i |m)^\infty = \phi_{2n,\infty} \left( c^\pm_i |m)^{2n} \right), \hbox{ where }
|m)^{2n} = \phi_{2n,\infty}^{-1}\left( |m)^\infty \right).
\label{c-inf}
\eeq
\end{defi}
Note that by Proposition~\ref{prop10}, the value chosen for $2n$ does not play a role, as long as $2n>\max\{2s,\rho(i)\}$.
Indeed, replacing $2n$ by $2n+2$ would give the same action in~\eqref{c-inf}.
Eq.~\eqref{c-inf} really says the following: in order to act with $c^\pm_i$ on an infinite pattern, first truncate this pattern appropriately, then act with $c^\pm_i$ on the truncated pattern, and finally extend all finite patterns thus obtained again to infinite patterns. 

\begin{theo}
The vector space $V(p)$, with basis vectors given by~\eqref{mni}, on which the action of the $\B(\infty,\infty)$ generators $c_i^\pm$ ($i\in\Z^*$) is defined by~\eqref{c-inf}, is an irreducible unitary Fock representation of $\B(\infty,\infty)$.
\end{theo}
The proof uses stability, and the fact that a finite set of generators $c_i^\pm$ ($i\in[\bar{n},n]^*$) satisfies the defining triple relations when acting on a truncation of stable GZ-patterns.
\begin{proof}
Since $\B(\infty,\infty)$ is generated by the elements $c_i^\pm$ ($i\in\Z^*$), $V(p)$ is a representation if we show that the action of the defining triple relation~\eqref{paraospi} on any vector $|m)^\infty$ of $V(p)$ is valid.
Let $|m)^\infty$ be a vector of $V(p)$, with stability index $2s$.
For any $\xi, \eta, \epsilon\in\{-1,+1\}$ and any $j,k,l\in\Z^*$, let
\beq
A^{\xi,\eta,\epsilon}_{j,k,l} = 
\lb\lb c_{ j}^{\xi}, c_{ k}^{\eta}\rb , c_{l}^{\epsilon}\rb + 2
\delta_{jl}\delta_{\epsilon, -\xi}\epsilon^{\langle l \rangle} 
(-1)^{\langle k \rangle \langle l \rangle }
c_{k}^{\eta} - 2 \epsilon^{\langle l \rangle }
\delta_{kl}\delta_{\epsilon, -\eta} c_{j}^{\xi}.
\label{Ajkl}
\eeq
We need to show that $A^{\xi,\eta,\epsilon}_{j,k,l} |m)^\infty = 0$.
In the given situation, let $n$ be such that
\[
2n \geq \max (2s, \rho(j),\rho(k), \rho(l) \} +6.
\]
The additional 6 is chosen because the action of $A^{\xi,\eta,\epsilon}_{j,k,l}$ involves 3 actions of $c_i^\pm$, and by Proposition~\ref{prop8} every such action could increase the stability index in the resulting vector by 2.
The vector $|m)^{2n} = \phi_{2n,\infty}^{-1}\left(|m)^\infty\right)$ is a basis vector of the corresponding $\B(n,n)$ Fock representation $V(p,n)$, and by construction of $n$ and the fact that $V(p,n)$ is a representation, the following holds:
\beq
A^{\xi,\eta,\epsilon}_{j,k,l} |m)^{2n} = 0.
\label{An}
\eeq
Think of the left hand side of~\eqref{An} as a linear combination of GZ-vectors with $2n$ rows.
By construction and by Propositions~\ref{prop8} and~\ref{prop9}, all these vectors are stable with respect to row $2n$.
Hence one can apply $\phi_{2n,\infty}^{-1}$, leading to
\beq
A^{\xi,\eta,\epsilon}_{j,k,l} |m)^{\infty} = 0.
\label{Ai}
\eeq
So we are dealing with a representation.
The inner product on $V(p)$ is defined by ${}^\infty(m'|m)^\infty =\delta_{m,m'}$.
With this, the representation $V(p)$ is unitary, since the matrix elements of $c^-_i$ and $c^+_i$ are related by
\beq
{}^\infty(m'| c^+_i |m)^\infty = {}^\infty(m| c^-_i |m')^\infty,
\eeq
by a straightforward extension of~\eqref{herm} and using the same stability argument as above.

The basis vector $|0\rangle =|0)^\infty$ (a zero GZ-pattern) clearly satisfies
\beq
\langle 0|0\rangle=1, \qquad c_{j}^- |0\rangle = 0, 
 \qquad
\lb c_j^-, c_k^+ \rb |0\rangle = p\delta_{jk}\, |0\rangle \quad
(j,k\in\Z^*)
\eeq
so we are dealing with a Fock representation.

$V(p)$ is generated by $|0\rangle$ and the action of the $c^\pm_i$'s.
Indeed, let $|m)^\infty$ be any vector of $V(p)$, stable with respect to row $2s$.
Let $|m)^{2s}= \phi_{2s,\infty}\left(|m)^\infty\right)$, then $|m)^{2s}$ is an element of the Fock space $V(p,s)$ of $\B(s,s)$,
since $m_{-s,2s}\leq p$.
Therefore, there must exist a polynomial expression $A$, say of degree~$k$, of products of the generators $c^\pm_i$ ($i\in[\bar{s},s]$) such that
\[
|m)^{2s} = A |0)^{2s}.
\]
In order to apply $\phi_{2s,\infty}$, we must be sure that the actions of $\phi_{2s,\infty}$ and $A$ in $V(p,s)$ commute.
Since each $c_i^\pm$ appearing in $A$ could increase the stability index by~2, let $2n=2s+2k$.
Obviously, also in $V(p,n)$ one hase
\[
|m)^{2n} = A |0)^{2n}.
\]
Now one can apply $\phi_{2n,\infty}^{-1}$, yielding $|m)^\infty = A |0\rangle$. 
But then also irreducibility in $V(p)$ follows, since ${}^\infty(m|A|0)^\infty=1$ implies ${}^\infty(0|A^\dagger|m)^\infty=1$, in other words, one can ``return to the vacuum vector $0\rangle$'' from any vector $|m)^\infty$.
\end{proof}

\section{Conclusion} 
\renewcommand{\theequation}{\Alph{section}.\arabic{equation}}
\setcounter{equation}{0}

For many years, the description of parastatistics Fock spaces with an infinite number of parafermions and parabosons was one of our ultimate goals.
In the past, we had already managed to extend our results on Fock spaces for a finite number of parabosons to the case of an infinite number of parabosons~\cite{SV2009}.
At the same time, we could extend our construction of Fock spaces for a finite number of parafermions to that of an infinite number~\cite{SV2009}.
In both cases, the extension to an algebra with an infinite number of creation and annihilation operators as generators turned out to be rather natural.
Also the construction of Fock spaces for the infinite rank case went without much difficulties.

Hence, after completing the description of Fock spaces for a combined system of $m$ parafermions and $n$ parabosons in~\cite{SV2015}, our hope was that the limit process to a combined infinite set would also be quite natural.
Unfortunately, this was not the case.
The Gelfand-Zetlin (GZ) basis of these Fock spaces had no natural limit for $m$ and $n$ going to infinity, essentially because the two triangular GZ patterns (related to those of $\gl(m)$ and $\gl(n)$) cannot grow at the same time.
So we had to rethink the situation.

The solution has been presented in this paper.
The main ingredient is a new GZ basis for the Fock spaces of the finite case, but this time for a combined system of $n$ parafermions and $n$ parabosons.
This GZ basis is referred to as the odd basis, a terminology introduced for $\gl(n|n)$ in~\cite{SV2016}.
The peculiar feature of this new basis is that the patterns consist of two opposite but entangled triangular GZ arrays.
In this new basis it becomes again quite natural to investigate what happens when $n$ grows to infinity, because the new GZ patterns have the right combinatorial properties.
This led us to the introduction of infinite row-stable GZ-patterns, of which we showed that they constitute a basis for the Fock spaces of a combined system of an infinite number of parafermions and parabosons.
In doing so, we also introduced an interesting new matrix form of the orthosymplectic Lie superalgebra $\B(n,n)$, which generalizes naturally to the infinite rank Lie superalgebra $\B(\infty,\infty)$.

\appendix 
\section{Appendix} \label{A}
\renewcommand{\theequation}{\Alph{section}.\arabic{equation}}
\setcounter{equation}{0}

The purpose of this appendix is to give the Clebsch-Gordan coefficients of $\mathfrak{gl}(n|n)$ corresponding to the tensor product $W([1,0,\ldots,0]) \otimes W([m]^{2n})$ of the $2n$-dimensional standard representation $W([1,0,\ldots,0])$ of $\mathfrak{gl}(n|n)$ with any $\mathfrak{gl}(n|n)$ covariant representation $W([m]^{2n})$.
Such CGC's have been computed in~\cite{CGC}, however in the standard GZ-basis.
Here we need these CGC's in the ``odd GZ-basis'', and that implies new calculations, which are quite involved.

We will briefly describe the technique used to compute these new CGC's. 
As explained in the main text, the following tensor product is valid:
\begin{equation}
W([1,0,\ldots,0]) \otimes W([m]^{2n}) =\bigoplus_{k\in[-n,n]^*} W([m]_{+(k)}^{2n}). \label{Atensprod}
\end{equation}
One can thus formally write down two orthonormal bases in the space~(\ref{Atensprod}): 
\beq
|1_j) \otimes \left| \begin{array}{l} [m]^{2n} \\[2mm] |m)^{2n-1} \end{array} \right)
\in W([1,0,\ldots,0]) \otimes W([m]^{2n}) 
\eeq
and
\beq
\left| \begin{array}{l} [m]^{2n}_{+(k)} \\[2mm] |m^\prime)^{2n-1} \end{array} \right)\in W([m]_{+(k)}^{2n}), \quad
(k\in[-n,n]^*),
\eeq
where the GZ-patterns appearing here satisfy the conditions~(\ref{cond3}), and $|1_j)$ ($j\in[1,2n]$) is a GZ-pattern consisting of only zeros in rows $1,2,\ldots,j-1$ (as usual, counted from the bottom), and rows of the form $1 0 \cdots 0$ in rows $j, j+1,\ldots 2n$. A zero row will be denoted by $0\cdots0=\dot 0$, and a row of the form $1 0 \cdots 0$ by $1\dot 0$.
The coefficients relating the two bases are the CGC's:

Then in general 
\begin{equation}
\left| \begin{array}{l} {[m]^{2n}_{+(k)}} \\[2mm] { |m^\prime)^{2n-1} }\end{array} \right)=
\sum 
  \left( \begin{array}{c}1 0 \cdots 0 0\\[-1mm]
1 0 \cdots 0\\[-1mm] \cdots\\[-1mm] 0 \end{array}; \right. 
	\begin{array}{ll} [m]^{2n} \\[2mm] |m)^{2n-1} \end{array}  
\left| \begin{array}{ll} [m]^{2n}_{+(k)} \\[2mm] |m^\prime)^{2n-1} \end{array} \right)
|1_j) \otimes \left| \begin{array}{l} [m]^{2n} \\[2mm] |m)^{2n-1} \end{array} \right),
\label{eq20}
\end{equation}
where the sum is over all possible values of $j$ and all allowed values of the GZ-labels in the pattern $|m)^{2n-1}$.
It will be useful to denote this CGC also as
\[
\left( \begin{array}{c}1 0 \cdots 0 0\\[-1mm]
1 0 \cdots 0\\[-1mm] \cdots\\[-1mm] 0 \end{array}; \right. 
	\begin{array}{ll} [m]^{2n} \\[2mm] |m)^{2n-1} \end{array}  
\left| \begin{array}{ll} [m]^{2n}_{+(k)} \\[2mm] |m^\prime)^{2n-1} \end{array} \right)
\equiv 
\left(\begin{array}{l}\\[-1mm]
|1_j)\\[-1mm] \\ \end{array}; \right. 
\begin{array}{ll} [m]^{2n} \\[2mm] |m)^{2n-1} \end{array} 
\left| \begin{array}{ll} [m]^{2n}_{+(k)} \\[2mm] |m^\prime)^{2n-1} \end{array} \right) .
\]
Using the actions of~\eqref{E-ii}-\eqref{Eii} on both sides of~\eqref{eq20}, it is clear that GZ-patterns of 
$\left| \begin{array}{l} {[m]^{2n}_{+(k)}} \\[2mm] { |m^\prime)^{2n-1} }\end{array} \right)$
and
$\left| \begin{array}{l} [m]^{2n} \\[2mm] |m)^{2n-1} \end{array} \right)$
are related in the following way:
\begin{itemize}
\item
there is no change in rows $1,2,\ldots,j-1$, i.e.\  $[m']^i = [m]^i$ for $i\in[1,j-1]$; 
\item
in rows $j, j+1,\ldots 2n$ there is a change by one unit for just one label, i.e.\ for $i\in[j,2n]$, $[m']^i$ is obtained from $[m]^i$ by $m'_{s,i}=m_{s,i}+1$ for one particular index $s$.
\end{itemize}

Let us briefly describe the procedure that we followed in order to compute the CGC's for the tensor product~\eqref{Atensprod}.
First of all, recall that the representations of a Lie superalgebra like $\gl(n|n)$ are also $\Z_2$ graded, and we choose the highest weight vectors $w_1\in W([1,0,\ldots,0])$ and $w_2\in W([m]^{2n})$ to be even.
Then it is clear that the highest weight vector $w\in W([m]^{2n}_{+(-n)})$ in the right hand side of~\eqref{Atensprod} is simply $w=w_1\otimes w_2$.
Using the known action of the $\gl(n|n)$ ``lowering'' generators $E_{ij}$ on $w$ from~\cite{SV2016}, one can construct the other weight vectors of $W([m]^{2n}_{+(-n)})$.
Acting by these generators also on vectors of the left hand side of~\eqref{Atensprod}, one can identify vectors on the left and right, and deduce the CGC's of the form~\eqref{eq20} with $k=-n$.
It is important to use the following rule in the actions on tensor products:
\beq
E_{ij} (v_1\otimes v_2) = (E_{ij} v_1) \otimes v_2 + (-1)^{\deg(E_{ij})\deg(v_1)} v_1 \otimes (E_{ij} v_2).
\eeq
Having worked in this way through $W([m]^{2n}_{+(-n)})$, one should next identify the highest weight vector of the second representation in the right hand side of~\eqref{Atensprod}, i.e.\ $w'\in W([m]^{2n}_{+(-n+1)})$.
Since the weight $\mu$ of $w'$ is known, one can construct the most general vector of $W([1,0,\ldots,0]) \otimes W([m]^{2n})$ with the same weight $\mu$, and express that it is orthogonal to each vector of weight $\mu$ belonging to $W([m]^{2n}_{+(-n)})$.
This vector is unique, up to a phase (which is chosen appropriately).
Then one can proceed to the other vectors of $W([m]^{2n}_{+(-n+1)})$, and get all CGC's of the form~\eqref{eq20} with $k=-n+1$.
Continuing in this way, explicit expression of the CGC's can be computed.

For the final formulas of these CGC's, it is useful to split them up in isoscalar factors and a CGC of a lower rank Lie superalgebra.
This is done as follows:
\begin{align}
& \left( \begin{array}{c}1 0 \cdots 0 0\\[-1mm]
1 0 \cdots 0\\[-1mm] \cdots\\[-1mm] 0 \end{array}; \right. 
	\begin{array}{ll} [m]^{2n} \\[2mm] |m)^{2n-1} \end{array}  
\left| \begin{array}{ll} [m]^{2n}_{+(k)} \\[2mm] |m^\prime)^{2n-1} \end{array} \right) \nn\\
&  = 
\left( \begin{array}{l} 1\dot{0} \\ \epsilon \dot{0} \end{array} \right.
\left|\begin{array}{l}  [m]^{2n} \\ {[m]}^{2n-1} \end{array} \right|
\left. \begin{array}{l} [m]^{2n}_{+(k)}  \\ {[m^\prime]}^{2n-1} \end{array} \right) \times
\left( \begin{array}{c}1 0 \cdots 0 0\\[-1mm]
1 0 \cdots 0\\[-1mm] \cdots\\[-1mm] 0 \end{array}; \right. 
	\begin{array}{ll} [m]^{2n-1} \\[2mm] |m)^{2n-2} \end{array}  
\left| \begin{array}{ll} [m^\prime]^{2n-1}_{+(k)} \\[2mm] |m^\prime)^{2n-2} \end{array} \right)
\label{IsoCGC}. 
\end{align}
In the right hand side, the first factor is an isoscalar factor~\cite{Vilenkin}, and the second factor is a CGC of $\mathfrak{gl}(n|n-1)$.
The middle pattern in the $\mathfrak{gl}(n|n-1)$ CGC is that of the $\mathfrak{gl}(n|n)$ CGC with the first row deleted.
The pattern in the isoscalar factor consists of the first two rows of the pattern in the left hand side, so $\epsilon$ is 0 or 1. 
If $\epsilon=0$, then $[m^\prime]^{2n-1}=[m]^{2n-1}$.
If $\epsilon=1$ then $[m^\prime]^{2n-1}=[m_{-n,2n-1},\ldots, m_{s,2n-1}+1,\ldots,m_{n-1,n-1}]=[m]^{2n-1}_{+s}$
for some $s$-value. 

We present now the results of our computation.

\bigskip
\noindent
{\bf Theorem:} 
The Clebsch-Gordan coefficients corresponding to the tensor product 
\[
W([1,0,\ldots,0]) \otimes W([m]^{2n})
\] 
of the standard $(2n)$-dimensional $\mathfrak{gl}(n|n)$  representation $W([1,0,\ldots,0])$ with an irreducible $\mathfrak{gl}(n|n)$ 
covariant tensor module $W([m]^{2n})$ are products of isoscalar factors:\\
for $j=1,2,\ldots,2n-1$: 
\begin{align}
& \left( \begin{array}{l}\\[-1mm]
|1_j)\\[-1mm] \\ \end{array};
\begin{array}{ll} [m]^{2n} \\[2mm] |m)^{2n-1} \end{array}  \right.
 \left| \begin{array}{ll} [m]^{2n}_{+(k)} \\[2mm] |m^\prime)^{2n-1} \end{array} \right)
 =  
\left( \begin{array}{l} 1\dot{0} \\ 1 \dot{0} \end{array}
\left| \begin{array}{l}  [m]^{2n} \\ {[m]}^{2n-1} \end{array} \right.\right|
\left. \begin{array}{l} [m]^{2n}_{+(k)}  \\ {[m^\prime]}^{2n-1} \end{array} \right)\times \ldots \nonumber\\
& \times\left( \begin{array}{l} 1\dot{0} \\ 1 \dot{0} \end{array} 
\left| \begin{array}{l}  [m]^{j+1} \\ {[m]}^{j} \end{array} \right.\right|
\left. \begin{array}{l} [m^\prime]^{j+1}  \\ {[m^\prime]}^{j} \end{array} \right) 
\left( \begin{array}{l} 1\dot{0} \\ 0 \dot{0} \end{array}
\left| \begin{array}{l}  [m]^{j} \\ {[m]}^{j-1} \end{array} \right. \right|
\left. \begin{array}{l} [m^\prime]^{j}  \\ {[m]}^{j-1} \end{array} \right)\times 1;   
 \label{CGC1}
\end{align} 
for $j=2n$:
\begin{align}
&  \left( \begin{array}{l}\\[-1mm]
|1_{2n})\\[-1mm] \\ \end{array};
\begin{array}{ll} [m]^{2n} \\[2mm] |m)^{2n-1} \end{array}  \right.
 \left| \begin{array}{ll} [m]^{2n}_{+(k)} \\[2mm] |m)^{2n-1} \end{array} \right)
 =  (-1)^{\sum_{i=1}^n\sum_{j=-i}^{-1}\theta_{j,2i-1}+\sum_{i=1}^{n-1}\sum_{j=1}^{i}\theta_{j,2i}}
\nonumber\\
&  \times
\left( \begin{array}{l} 1\dot{0} \\ 0 \dot{0} \end{array}
\left| \begin{array}{l}  [m]^{2n} \\ {[m]}^{2n-1} \end{array} \right. \right|
\left. \begin{array}{l} [m]^{2n}_{+(k)}  \\ {[m]}^{2n-1} \end{array} \right).  
 \label{CGC10}
\end{align} 

\noindent
In the right hand side of~(\ref{CGC1})-(\ref{CGC10}),
\[
\left( \begin{array}{l} 1 \dot{0} \\ \epsilon \dot{0} \end{array}
\left|\begin{array}{l} [m]^{t} \\ {[m]}^{t-1} \end{array} \right.  \right|
\left. \begin{array}{l} [m]^{t}_{+(k)} \\ {[m']}^{t-1} \end{array} \right) \qquad (\epsilon\in\{0,1\})
\]
is a $\mathfrak{gl}(t|t)\supset\mathfrak{gl}(t|t-1)$ ($1\leq t\leq n$) or
 a $\mathfrak{gl}(t|t-1)\supset\mathfrak{gl}(t-1|t-1)$ ($2\leq t\leq n$) isoscalar factor following the chain of subalgebras
\begin{equation}
\mathfrak{gl}(n|n) \supset \mathfrak{gl}(n|n-1) \supset \mathfrak{gl}(n-1|n-1) \supset \mathfrak{gl}(n-1|n-2) \supset \mathfrak{gl}(n-2|n-2) \supset 
\cdots \supset \mathfrak{gl}(1|1) \supset \mathfrak{gl}(1|0)\equiv \mathfrak{gl}(1) .
\label{chain-nn}
\end{equation}

\noindent
For the $\mathfrak{gl}(t|t)\supset \mathfrak{gl}(t|t-1)$ isoscalar factors there are six different expressions,
depending on the position of the pattern changes in the right hand side.
For all these expressions, it is common to use other labels than the $m_{ij}$'s:
\begin{equation}
l_{i,2k-\varphi}=m_{i,2k-\varphi}-i \; (-k\leq i \leq -1);\;
l_{s,2k-\varphi}=-m_{s,2k-\varphi}+s \;(1\leq s\leq k-\varphi), 
\; \varphi =0,1; \; k=1,\ldots,n.
\label{lir}
\end{equation}
Furthermore, $\prod_{i=a(\neq k)}^{b}$ means the product over all $i$-values
running from $a$ to $b$, but excluding $i=k$. \\
These six expressions are given by~\eqref{liso1}-\eqref{liso6}:
\begin{align}
 &\left( \begin{array}{l} 1 \dot{0} \\0 \dot{0} \end{array}
\left|\begin{array}{l} [m]^{2t} \\ {[m]}^{2t-1} \end{array} \right.  \right|
\left. \begin{array}{l} [m]^{2t}_{+(k)} \\ {[m]}^{2t-1} \end{array} \right) 
= (-1)^{t+k}(1-\theta_{k,2t-1})(-1)^{{\sum_{i=k}^{-1}}\theta_{i,2t-1}}\nonumber\\
 &
\times
\left( \prod_{i=-t(\neq k)}^{-1} \left(\frac{l_{k,2t}-l_{i,2t}+1}{l_{k,2t}-l_{i,2t-1}}
\right)  
 \frac{\prod_{s=1}^{t-1}  
(l_{k,2t}-l_{s,2t-1} )}
{ \prod_{s=1}^{t} (l_{k,2t}-l_{s,2t}+1)}
\right)^{1/2} \quad(1\leq t\leq n; \; -t\leq k \leq -1);\label{liso1}
\end{align}

\begin{align}
\left( \begin{array}{l} 1 \dot{0} \\0 \dot{0} \end{array}
\left| \begin{array}{l} [m]^{2t} \\ {[m]}^{2t-1} \end{array} \right. \right|
\left. \begin{array}{l} [m]^{2t}_{+(k)} \\ {[m]}^{2t-1} \end{array} \right)
&= \left( \prod_{i=-t}^{-1} \left(\frac{l_{i,2t}-l_{k,2t}}{l_{i,2t-1}-l_{k,2t}+1}\right)  
 \frac{\prod_{s=1}^{t-1}  (l_{s,2t-1}-l_{k,2t}+1 )}{ \prod_{s=1(\neq k)}^{t} (l_{s,2t}-l_{k,2t})} \right)^{1/2} \nonumber\\
& \hskip 4cm \quad (2\leq t \leq n;\;1\leq k \leq t) 
\label{liso2} \\
& = (-1)^{\theta_{11}}\left( \frac{l_{-1,2}-l_{1,2}}{l_{-1,1}-l_{1,2}+1}
\right)^{1/2} \qquad(t=k=1);\nonumber
\end{align}

\begin{align}
&\left( \begin{array}{l} 1 \dot{0} \\1 \dot{0} \end{array} 
\left| \begin{array}{l} [m]^{2t} \\ {[m]}^{2t-1} \end{array} \right.\right|
\left. \begin{array}{l} [m]^{2t}_{+(k)} \\ {[m]}^{2t-1}_{+(q)} \end{array} \right)
=(-1)^{k+q}(-1)^{{\sum_{i=\min(k+1,q+1)}^{\max(k-1,q-1)}}\theta_{i,2t-1}}(\delta_{kq}+(1-\delta_{kq})\theta_{q,2t-1}(1-\theta_{k,2t-1})) \nonumber\\
&\times \left( \prod_{i=-t(\neq k,q)}^{-1} 
\frac{(l_{i,2t-1}-l_{k,2t-1}-1-\delta_{kq}+2\theta_{i,2t-1})
(l_{i,2t-1}-l_{q,2t-1})}{(l_{i,2t}-l_{k,2t})(l_{i,2t}-l_{q,2t})}
\right)^{{\theta_{q,2t-1}}/2} \nonumber\\  
&
\times\frac{1}{(l_{k,2t}-l_{q,2t})^{1-\delta_{kq}}}
\left( \prod_{s=1}^{t} \left(
\frac{l_{q,2t}-l_{s,2t}}
{l_{k,2t}-l_{s,2t}+1}
\right)
\prod_{s=1}^{t-1} \left(
\frac{l_{k,2t}-l_{s,2t-1}}
{l_{q,2t-1}-l_{s,2t-1}}
\right)
\right)^{{\theta_{q,2t-1}}/2}
\quad (-t\leq k,q \leq -1);\label{liso3}
\end{align} 

\begin{align}
&\left( \begin{array}{l} 1 \dot{0} \\1 \dot{0} \end{array} 
\left| \begin{array}{l} [m]^{2t} \\ {[m]}^{2t-1} \end{array} \right. \right|
\left. \begin{array}{l} [m]^{2t}_{+(k)} \\ {[m]}^{2t-1}_{+(q)} \end{array} \right)
=(-1)^{k+t+1}(-1)^{{\sum_{i=-t}^{k-1}}\theta_{i,2t-1}}(1-\theta_{k,2t-1})
\left(\frac{1}{l_{k,2t}-l_{q,2t-1}}\right)^{1/2}
\nonumber\\
&\times \left( \prod_{i=-t(\neq k)}^{-1} \left(\frac{(l_{i,2t-1}-l_{k,2t-1}-1+2\theta_{i,2t-1})
(l_{i,2t-1}-l_{q,2t-1}+1)}{(l_{i,2t}-l_{k,2t})(l_{i,2t}-l_{q,2t-1})}
\right)\right)^{1/2} \nonumber\\  
& \times
\left( \prod_{s=1}^{t} \left(\frac{|l_{s,2t}-l_{q,2t-1}|}
{(l_{k,2t}-l_{s,2t}+1)}
\right) \prod_{s=1(\neq q)}^{t-1} \left(\frac{l_{k,2t}-l_{s,2t-1}}
{|l_{s,2t-1}-l_{q,2t-1}+1|}
\right)\right)^{1/2}\nonumber\\
& \hskip 8cm (-t\leq k \leq -1, \quad 1\leq q \leq t-1);\label{liso4}
\end{align}

\begin{align}
&\left( \begin{array}{l} 1 \dot{0} \\1 \dot{0} \end{array} 
\left| \begin{array}{l} [m]^{2t} \\ {[m]}^{2t-1} \end{array} \right. \right|
\left. \begin{array}{l} [m]^{2t}_{+(k)} \\ {[m]}^{2t-1}_{+(q)} \end{array} \right)
=(-1)^{q+t+1}(-1)^{{\sum_{i=q+1}^{-1}}\theta_{i,2t-1}}\theta_{q,2t-1}
\left(\frac{1}{l_{q,2t}-l_{k,2t}+1}\right)^{1/2}
\nonumber\\
&\times \left( \prod_{i=-t}^{-1} \left(\frac{l_{i,2t}-l_{k,2t}}{l_{i,2t-1}-l_{k,2t}+1}
\right) \prod_{i=-t(\neq q)}^{-1} \Big|\frac{l_{q,2t-1}-l_{i,2t-1}}{l_{q,2t}-l_{i,2t}}
 \Big| 
 \prod_{s=1(\neq k)}^{t} \Big|\frac{l_{q,2t}-l_{s,2t}}{l_{s,2t}-l_{k,2t}}
\Big|
 \prod_{s=1}^{t-1} \Big|\frac{l_{s,2t-1}-l_{k,2t}+1}{l_{q,2t}-l_{s,2t-1}-1}
 \Big|\right)^{1/2}\nonumber\\
& \hskip 8cm (1\leq k \leq t, \quad -t\leq q \leq -1);\label{liso5}
\end{align}

\begin{align} 
&\left( \begin{array}{l} 1 \dot{0} \\1 \dot{0} \end{array} 
\left| \begin{array}{l} [m]^{2t} \\ {[m]}^{2t-1} \end{array} \right.\right|
\left. \begin{array}{l} [m]^{2t}_{+(k)} \\ {[m]}^{2t-1}_{+(q)} \end{array} \right)
=S(k,q) (-1)^{{\sum_{i=-t}^{-1}}\theta_{i,2t-1}}\left( \prod_{i=-t}^{-1} \left(\frac{(l_{i,2t}-l_{k,2t})
(l_{i,2t-1}-l_{q,2t-1}+1)}{(l_{i,2t-1}-l_{k,2t}+1)(l_{i,2t}-l_{q,2t-1})}
\right)\right)^{1/2} \nonumber\\ 
& \times
\left( \prod_{s=1(\neq k)}^{t} \Big|\frac{l_{s,2t}-l_{q,2t-1}}
{l_{s,2t}-l_{k,2t}}\Big|
 \prod_{s=1(\neq q)}^{t-1} \Big| \frac{l_{s,2t-1}-l_{k,2t}+1}
{l_{s,2t-1}-l_{q,2t-1}+1}
\Big| \right)^{1/2} 
\quad(1\leq k \leq t, \quad 1\leq q \leq t-1).\label{liso6}
\end{align}

\noindent
Also for the $\mathfrak{gl}(t|t-1)\supset \mathfrak{gl}(t-1|t-1)$ isoscalar factors there are six different expressions,
depending on the position of the pattern changes in the right hand side. 
These six expressions are given by~\eqref{Sliso1}-\eqref{Sliso6}:

\begin{align}
&\left( \begin{array}{l} 1 \dot{0} \\\;\; \dot{0} \end{array}
\left|\begin{array}{l} [m]^{2t-1} \\ {[m]}^{2t-2} \end{array} \right.  \right|
\left. \begin{array}{l} [m]^{2t-1}_{+(k)} \\ {[m]}^{2t-2} \end{array} \right) 
\nonumber\\
&
= (-1)^{t+k} \left( \frac{\prod_{s=-t+1}^{-1}  
(l_{k,2t-1}-l_{s,2t-2} )}
{ \prod_{s=-t(\neq k)}^{-1} (l_{k,2t-1}-l_{s,2t-1})}
 \prod_{i=1}^{t-1} \left(\frac{l_{k,2t-1}-l_{i,2t-1}}{l_{k,2t-1}-l_{i,2t-2}}
\right) 
\right)^{1/2} \quad(-t\leq k \leq -1);\label{Sliso1}
\end{align}

\begin{align}
&\left( \begin{array}{l} 1 \dot{0} \\\;\; \dot{0} \end{array}
\left| \begin{array}{l} [m]^{2t-1} \\ {[m]}^{2t-2} \end{array} \right. \right|
\left. \begin{array}{l} [m]^{2t-1}_{+(k)} \\ {[m]}^{2t-2} \end{array} \right)
= \theta_{k, 2t-2}(-1)^k(-1)^{\sum_{i=k+1}^{t-1}\theta_{i,2t-2}}\nonumber\\
&
\times \left( \frac{\prod_{s=-t+1}^{-1}  
(l_{s,2t-2}-l_{k,2t-1}+1 )}
{ \prod_{s=-t}^{-1} (l_{s,2t-1}-l_{k,2t-1}+1)}
 \prod_{i=1(\neq k)}^{t-1} \left(\frac{l_{i,2t-1}-l_{k,2t-1}+1}{l_{i,2t-2}-l_{k,2t-1}+1}
\right) 
\right)^{1/2} \quad(1\leq k \leq t-1);
\label{Sliso2}
\end{align}

\begin{align}
&\left( \begin{array}{l} 1 0 \dot{0} \\\;\; 1 \dot{0} \end{array} 
\left| \begin{array}{l} [m]^{2t-1} \\ {[m]}^{2t-2} \end{array} \right.\right|
\left. \begin{array}{l} [m]^{2t-1}_{+(k)} \\ {[m]}^{2t-2}_{+(q)} \end{array} \right)
\nonumber\\
&
=(-1)^{k+q}S(-k,-q)
\left( \prod_{i=-t(\neq k)}^{-1} \left(
\frac{l_{q,2t-2}-l_{i,2t-1}+1}
{l_{k,2t-1}-l_{i,2t-1}}
\right)
\prod_{i=-t+1(\neq q)}^{-1} \left(
\frac{l_{k,2t-1}-l_{i,2t-2}}
{l_{q,2t-2}-l_{i,2t-2}+1}
\right)
\right)^{1/2}\nonumber\\
&\times \left( \prod_{s=1}^{t-1} 
\frac{(l_{k,2t-1}-l_{s,2t-1})
(l_{q,2t-2}-l_{s,2t-2}+1)}{(l_{k,2t-1}-l_{s,2t-2})(l_{q,2t-2}-l_{s,2t-1}+1)}
\right)^{1/2} 
\quad (-t\leq k\leq -1, \;-t+1\leq q\leq -1) ;\label{Sliso3}
\end{align}

\begin{align}
&\left( \begin{array}{l} 1 0 \dot{0} \\\;\; 1 \dot{0} \end{array} 
\left| \begin{array}{l} [m]^{2t-1} \\ {[m]}^{2t-2} \end{array} \right. \right|
\left. \begin{array}{l} [m]^{2t-1}_{+(k)} \\ {[m]}^{2t-2}_{+(q)} \end{array} \right)
=(-1)^{k+t}(-1)^{{\sum_{i=1}^{q-1}}\theta_{i,2t-2}}(1-\theta_{q,2t-2})
\left(\frac{1}{l_{k,2t-1}-l_{q,2t-2}+1}\right)^{1/2}
\nonumber\\
&\times \left( \prod_{i=-t(\neq k)}^{-1} \left(\frac{l_{i,2t-1}-l_{q,2t-2}
}{l_{k,2t-1}-l_{i,2t-1}}
\right)   \prod_{i=-t+1}^{-1} \left(\frac{l_{k,2t-1}-l_{i,2t-2}
}{l_{i,2t-2}-l_{q,2t-2}}
\right)       \right)^{1/2} \nonumber\\  
& \times
\left(  \prod_{s=1(\neq q)}^{t-1} \left(\frac{(l_{k,2t-1}-l_{s,2t-1})(l_{s,2t-2}-l_{q,2t-2})}
{(l_{k,2t-1}-l_{s,2t-2})(l_{s,2t-1}-l_{q,2t-2})}
\right)\right)^{1/2} \quad (-t\leq k \leq -1, \quad 1\leq q \leq t-1);\label{Sliso4}
\end{align}

\begin{align}
&\left( \begin{array}{l} 1 0 \dot{0} \\\;\; 1 \dot{0} \end{array} 
\left| \begin{array}{l} [m]^{2t-1} \\ {[m]}^{2t-2} \end{array} \right. \right|
\left. \begin{array}{l} [m]^{2t-1}_{+(k)} \\ {[m]}^{2t-2}_{+(q)} \end{array} \right)
=(-1)^{k+q+t}(-1)^{{\sum_{i=k+1}^{t-1}}\theta_{i,2t-2}}\theta_{k,2t-2}
\left(\frac{1}{l_{q,2t-2}-l_{k,2t-1}+1}\right)^{1/2}
\nonumber\\
&\times \left( \prod_{i=-t}^{-1} \left(\frac{|l_{i,2t-1}-l_{q,2t-2}-1|}{l_{i,2t-1}-l_{k,2t-1}+1}
\right) \prod_{i=-t+1(\neq q)}^{-1}  \left(   \frac{l_{i,2t-2}-l_{k,2t-1}+1}{|l_{i,2t-2}-l_{q,2t-2}-1|}\right) \right)^{1/2}
\nonumber\\ 
&\times \left(\prod_{s=1(\neq k)}^{t-1} \left(\frac{(l_{q,2t-2}-l_{s,2t-2}+1)(l_{s,2t-1}-l_{k,2t-1}+1)}{(l_{q,2t-2}-l_{s,2t-1}+1)(l_{s,2t-2}-l_{k,2t-1}+1)}\right)
 \right)^{1/2} \quad  (1\leq k \leq t-1, \ -t+1\leq q \leq -1);\label{Sliso5}
\end{align}

\begin{align}
&\left( \begin{array}{l} 1 0 \dot{0} \\\;\; 1 \dot{0} \end{array} 
\left| \begin{array}{l} [m]^{2t-1} \\ {[m]}^{2t-2} \end{array} \right.\right|
\left. \begin{array}{l} [m]^{2t-1}_{+(k)} \\ {[m]}^{2t-2}_{+(q)} \end{array} \right)
=(-1)^{k-1}(-1)^{{\sum_{i=1}^{t-1}}\theta_{i,2t-2}}(\delta_{kq}+(1-\delta_{kq})\theta_{k,2t-2}(1-\theta_{q,2t-2})) \nonumber\\
&
\times\frac{1}{(l_{k,2t-1}-l_{q,2t-1})^{1-\delta_{kq}}}
\left( \prod_{i=-t}^{-1} \left(
\frac{l_{i,2t-1}-l_{q,2t-2}}
{l_{i,2t-1}-l_{k,2t-1}+1}
\right)
\prod_{i=-t+1}^{-1} \left(
\frac{l_{i,2t-2}-l_{k,2t-1}+1}
{l_{i,2t-2}-l_{q,2t-2}}
\right)
\right)^{{(1-\theta_{q,2t-2})}/2}
\nonumber\\  
&\times \left( \prod_{s=1(\neq k,q)}^{t-1} 
\frac{(l_{s,2t-2}-l_{q,2t-2}+1-\delta_{kq}+2\theta_{s,2t-2})
(l_{s,2t-2}-l_{q,2t-2})}{(l_{s,2t-1}-l_{k,2t-1})(l_{s,2t-1}-l_{q,2t-2})}
\right)^{{(1-\theta_{q,2t-2})}/2}
\ (1\leq k,q \leq t-1);\label{Sliso6}
\end{align} 

\noindent
In all of the formulas above, $\dot{0}$ stands for an appropriate sequence of zeros in the pattern, and we use   
\beq
S(k,q) = \left\{ \begin{array}{rcl}
 {1} &  \hbox{for} &  k\leq q  \\ 
 {-1} &  \hbox{for} &  k>q.
 \end{array}\right. .
\label{SS}
\eeq

\section*{Acknowledgments}
The authors were supported by the Joint Research Project ``Lie superalgebras - applications in quantum theory'' in the framework of an international collaboration programme between the Research Foundation -- Flanders (FWO) and the Bulgarian Academy of Sciences. 
N.I. Stoilova was partially supported by the Bulgarian National Science Fund, grant KP-06-N28/6, and J. Van der Jeugt was partially supported by the EOS Research Project 30889451.

\end{document}